%% file: main.tex
\documentclass[a4paper,UKenglish,cleveref, autoref, thm-restate]{lipics-v2021}

\usepackage{graphicx}
\usepackage{amsmath, nccmath}
\usepackage{amsfonts}
\usepackage{amssymb} %
\usepackage{amsthm}
\usepackage{mathtools}
\usepackage{color, ifthen} %
\usepackage[all]{xy} \SelectTips{cm}{}  %
\usepackage{quiver}
\usepackage{hyperref}
\usepackage{cleveref}

\newcommand{\Set}{\mathbf{Set}}
\newcommand{\Cat}{\mathbf{Cat}}

\newcommand{\id}{\mathrm{id}}

\newcommand{\ol}[1]{\overline{#1}}
\newcommand{\Coalg}[1]{\mathrm{Coalg}(#1)}
\newcommand{\CoalgOcov}[2]{\mathrm{Coalg}^\cat{O}_{#2}(#1)}
\newcommand{\CoalgOctr}[1]{\mathrm{Coalg}^\cat{O}(#1)_{\mathrm{ctr}}}

\newcommand{\cat}[1]{\mathbb{#1}}
\newcommand{\nat}{\mathbb{N}}

\newcommand{\pullbackmark}[2]{\save ;p+<.8pc,0pc>:(0,-1)::%
(#1) *{\phantom{Z}} %
;p+(#2)-(0,0) **@{-}%
;p-(#1)+(0,0) *{\phantom{Z}} **@{-} \restore}

\newcommand{\bset}{\mathbf{B}}

\newcommand{\btrue}{\mathrm{t}}
\newcommand{\bfalse}{\mathrm{f}}

\newcommand{\Vc}[1]{V^*_{#1}}
\newcommand{\last}{\mathrm{last}}

\newcommand{\Pol}{\mathsf{Sch}}

\newcommand{\Hist}{\mathsf{Hist}}

\newcommand{\Bel}{\mathsf{Bel}}

\newcommand{\ev}{\mathrm{ev}}
\newcommand{\str}{\mathrm{st}}
\newcommand{\stl}{\mathrm{st}} %
\newcommand{\stlift}{\widehat{\stl}} %

\newcommand{\obsv}{\mathrm{obs}}
\newcommand{\slicemonad}[2]{\ol{#1}}
\newcommand{\slicemonadO}[2]{\ol{#1}_{#2}}
\newcommand{\slicefunc}[2]{#1_{/#2}}
\newcommand{\domofslice}[3]{\dom\ol{#1}(#3)} %

\newcommand{\Sl}{\mathrm{Sl}}

\newcommand{\dom}{d}

\newcommand{\pfunc}{\mathcal{P}}
\newcommand{\pefunc}{\mathcal{P}_{\emptyset}}
\newcommand{\R}{\mathbb{R}_{\geq 0}}
\newcommand{\eR}{[0, \infty]}
\newcommand{\N}{\mathbb{N}}
\newcommand{\dfunc}{\mathcal{D}}

\newcommand{\rfunc}{\mathcal{R}}

\newcommand{\expTR}{\mathrm{TER}}
\newcommand{\supp}{\mathrm{supp}}

\usetikzlibrary{arrows.meta,automata,positioning,calc}

\newcommand\restr[2]{{%
  \left.\kern-\nulldelimiterspace %
  #1 %
  \vphantom{\big|} %
  \right|_{#2} %
  }}

\newif\ifarxiv\arxivtrue

 \newif\ifdraft\draftfalse

\ifdraft
\usepackage{todonotes}

\else
\usepackage[disable]{todonotes}
\fi

\newenvironment{exa}
  {\begin{example}}
  {\end{example}}

\newenvironment{defi}
  {\begin{definition}}
  {\end{definition}}

\newenvironment{lem}
  {\begin{lemma}}
  {\end{lemma}}

\newenvironment{rem}
{\begin{remark}}
{\end{remark}}

\newenvironment{prop}
{\begin{proposition}}
{\end{proposition}}

\newenvironment{thm}
{\begin{theorem}}
{\end{theorem}}

\newtheorem{assumption}{Assumption}

\newenvironment{asm}
{\begin{assumption}}
{\end{assumption}}
\crefformat{section}{\S#2#1#3}
\Crefformat{section}{\S#2#1#3}

\crefname{section}{\S}{\S\S}
\Crefname{section}{\S}{\S\S}

\crefname{definition}{Def.}{Definitions}
\Crefname{definition}{Def.}{Definitions}

\crefname{proposition}{Prop.}{Propositions}
\Crefname{proposition}{Prop.}{Propositions}

\crefname{theorem}{Thm.}{Theorems}
\Crefname{theorem}{Thm.}{Theorems}

\crefname{lemma}{Lem.}{Lemmas}
\Crefname{lemma}{Lem.}{Lemmas}

\crefname{figure}{Fig.}{Figures}
\Crefname{figure}{Fig.}{Figures}

\crefname{equation}{Eq.}{Equations}
\Crefname{equation}{Eq.}{Equations}

\title{From Coalgebraic Determinization\\ to Belief Construction for Partial Observability} %

\author{Mayuko Kori}{Research Institute for Mathematical Sciences, Kyoto University, Kyoto}{mkori@kurims.kyoto-u.ac.jp}{}{}%

\author{Kazuki Watanabe}{National Institute of Informatics, Tokyo\\ The Graduate University for Advanced Studies (SOKENDAI), Tokyo}{kazukiwatanabe@nii.ac.jp}{}{}%

\authorrunning{Mayuko Kori and Kazuki Watanabe} %

\Copyright{Mayuko Kori and Kazuki Watanabe} %


\category{} %

\acknowledgements{
We thank the anonymous reviewers for their constructive comments.
M.~K.~is supported by the JST grant No.~JPMJAX25CD, and
K.~W.~is supported by the JST grant No.~JPMJPR25KD.}%

\nolinenumbers %

\begin{document}

\maketitle

\begin{abstract}
  The \emph{belief construction} is a fundamental technique for transforming partially observable systems to fully observable ones while preserving the relevant semantics. It plays a central role in the analysis of partially observable systems, in particular
  \emph{partially observable
  Markov decision processes} (POMDPs), which is a central model in artificial intelligence and formal verification. 
  In this paper, we develop a coalgebraic framework for the belief construction. 
  To handle observations categorically, we lift a monad to slice categories and introduce a belief decomposition that reorganizes states according to their observations. 
  This allows us to introduce a coalgebraic generalization of the belief construction, obtained by combining the belief decomposition with the coalgebraic determinization of Silva, Bonchi, Bonsangue, and Rutten.
  In this framework, we show that the semantics of a partially observable system coincides with that of the corresponding belief coalgebra. 
  We then study when the latter further agrees with the semantics of its fully observable counterpart, and use this to identify conditions under which the semantics of a partially observable system coincides with that of the corresponding fully observable belief system. 
As a consequence, we recover the standard equivalence between POMDPs and belief MDPs, and obtain a new equivalence result for weighted transition systems with the semimodule monad. 

\end{abstract}

\section{Introduction}\label{S:one}
The notion of observation is ubiquitous in artificial intelligence and formal verification, since in real-world settings it is often unrealistic to assume that we have precise information about the systems under analysis. 
\emph{Partially observable Markov decision processes} (POMDPs) (e.g.~\cite{RN2020}) provide a central model for incorporating partial observability, and hence for reasoning about imperfect information about systems. 

\begin{figure}[t]
\centering

\begin{minipage}[t]{0.4\textwidth}
\centering
\scalebox{0.75}{
\begin{tikzpicture}[
  >=Latex,
  node distance=13mm,
  state/.style={draw, circle, minimum size=8mm, inner sep=0pt},
  every edge/.style={draw, ->, >=Latex}
]
\node[state] (s0) {$s_0$};
\node[state, below left=of s0, fill=cyan] (s1) {$s_1$};
\node[state, below right=of s0, fill=cyan] (s2) {$s_2$};
\node[state, below=of s1] (t1) {$t_1$};
\node[state, below=of s2] (t2) {$t_2$};

\draw[->] ([xshift=-10mm] s0.west) -- (s0.west);

\path
(s0) edge[bend left=10] node[left, yshift=3pt]  {\footnotesize$a, \frac12 \, / \, 1$} (s1)
(s0) edge[bend right=10] node[right, yshift=3pt] {\footnotesize$a, \frac12 \, / \, 1$} (s2)
(s0) edge[bend left=30]  node[right]  {\footnotesize$b, 1\,/\,0$} (t1);

\path
(s1) edge[bend left=8]  node[right]  {\footnotesize$a, 1\,/\,1$} (t1)
(s1) edge[bend right=8] node[left] {\footnotesize$b, 1\,/\,0$} (t1);

\path
(s2) edge[bend right=8] node[left] {\footnotesize$a, 1\,/\,0$} (t2)
(s2) edge[bend left=8]  node[right] {\footnotesize$b, 1\,/\,2$} (t2);

\path
(t1) edge[loop below] node {\footnotesize$a,b, 1\,/\,0$} ()
(t2) edge[loop below] node {\footnotesize$a,b, 1\,/\,0$} ();
\end{tikzpicture}
}
\caption{A POMDP $c$, where the states $s_1$ and $s_2$ are assigned the same observation $o$. }
\label{fig:pomdpLeftEx}

\end{minipage}
\hfill
\begin{minipage}[t]{0.5\textwidth}
\centering
\scalebox{0.75}{
\begin{tikzpicture}[
  >=Latex,
  node distance=13mm,
  bstate/.style={draw, rounded corners, minimum height=8mm, minimum width=10mm, align=center},
  every edge/.style={draw, ->, >=Latex}
]
\node[bstate] (b0) {$\delta_{s_0}$};
\node[bstate, right=of b0, fill=cyan] (mu) {$\mu$};
\node[bstate, above right=of mu] (bt1) {$\delta_{t_1}$};
\node[bstate, below right=of mu] (bt2) {$\delta_{t_2}$};

\draw[->] ([xshift=-10mm] b0.west) -- (b0.west);

\path
(b0) edge[] node[above] {\footnotesize$a, 1 \,/\, 1$} (mu)
(b0) edge[bend left=15] node[above] {\footnotesize$b, 1 \,/\, 0$} (bt1);

\path
(mu) edge[bend left=14]  node[above, left] {\footnotesize$a, \frac12\,/\,1$} (bt1)
(mu) edge[bend right=14] node[below, right] {\footnotesize$b, \frac12\,/\,0$} (bt1)
(mu) edge[bend left=14]  node[above, right] {\footnotesize$b, \frac12\,/\,2$} (bt2)
(mu) edge[bend right=14] node[below, left] {\footnotesize$a, \frac12\,/\,0$} (bt2);

\path
(bt1) edge[loop above] node {\footnotesize$a,b, 1\,/\, 0$} ()
(bt2) edge[loop below] node {\footnotesize$a,b, 1 \,/\, 0$} ();
\end{tikzpicture}
}
\caption{The reachable part of the belief MDP, where $\mu \coloneqq \tfrac12\delta_{s_1}+\tfrac12\delta_{s_2}$.}
\label{fig:beliefPOMDP}
\end{minipage}

\end{figure}

\begin{exa}
  We illustrate this with the POMDP $c$ shown in~\cref{fig:pomdpLeftEx}.\footnote{While rewards are assigned to transitions in the example, our coalgebraic framework for POMDPs assigns rewards to states.
This difference is not essential, in the sense that the former can be modelled by the latter by adding dummy states, and the converse direction is trivial.  }
A POMDP is a Markov decision process (MDP) equipped with an observation function assigning to each state an observation.
In this example, there are five states $S \coloneqq \{s_0,s_1,s_2,t_1, t_2\}$, and the observation function $\obsv\colon S\rightarrow \{o, o_0, o_1, o_2\}$ is given by 
\begin{align*}
  \obsv(s_0) \coloneqq o_0,\quad \obsv(s_1) = \obsv(s_2) \coloneqq o, \quad \obsv(t_1) \coloneqq o_1, \quad \obsv(t_2) \coloneqq o_2.
\end{align*}
From each state, choosing an action $a$ or $b$ induces a probabilistic transition together with an immediate reward.
For instance, from the initial state $s_0$, choosing action $a$ leads to $s_1$ with probability $\frac{1}{2}$ and reward $1$, and to $s_2$ with probability $\frac{1}{2}$ and reward $1$.

The quantity of interest here is 
the maximal total expected reward 
over observation-based schedulers.
An observation-based scheduler is 
a map
$u \colon O^{+}\to \{a,b\}$, where $O = \{o_0, o, o_1, o_2\}$
and $O^+$ is the set of nonempty finite sequences of observations,
which selects the next action based only on the history of observations, rather than the underlying states.\footnote{This presentation is equivalent to the more common presentation of schedulers as functions $O \times (\{a, b\} \times O)^* \to \{a,b\}$ since previously chosen actions are themselves determined by the preceding observations.}
For the POMDP $c$,
the maximal total expected reward from $s_0$ is
\begin{equation} \label{eq:ex_pomdp}
\max_{u\in \{a, b\}^{O^{+}}} \expTR_{c, u}  = \frac{1}{2}(1+2)+\frac{1}{2}(1+0)=2,
\end{equation}
where $\expTR_{c, u}$ is the total expected reward of $c$ under the scheduler $u$.
The maximum is attained by choosing action $a$ first and then action $b$.
Notably, if we instead regard the same system as a fully observable MDP, then the maximal total expected reward becomes $\frac{5}{2}$, since the choice of action can depend on the visited state ($s_1$ or $s_2$).
\end{exa}

The belief construction for POMDPs (e.g.~\cite{ShaniPK13}) transforms 
a POMDP into an equivalent fully observable MDP, called its \emph{belief MDP}. The states of this MDP are probability distributions over the original state space, called \emph{beliefs}.
Although the belief MDP of a finite-state POMDP is generally infinite-state, the construction reduces the original problem to one on MDPs, for which many effective abstraction and approximation techniques have been developed (e.g.~\cite{SimaoS023,BorkKQ22,AndriushchenkoBCJKM26,Norman0Z17,RoyGT05}).

\begin{exa}
  We present  in~\cref{fig:beliefPOMDP} the belief MDP $c^\Bel$ of the POMDP $c$ above;
we write $\delta_s$ for the Dirac distribution at a state $s$. 
The figure shows only the reachable part 
 from the initial belief $\delta_{s_0}$. 
The maximal total expected reward from the initial belief $\delta_{s_0}$ is
\begin{equation} \label{eq:ex_bel}
\frac{1}{2}(1+2) + \frac{1}{2}(1+0) = 2,
\end{equation}
which coincides with the maximal total expected reward~\eqref{eq:ex_pomdp} of the POMDP.
\end{exa}

The above example illustrates that the usual belief construction for POMDPs is semantics-preserving: the maximal total expected reward of the original POMDP coincides with that of its belief MDP.
This naturally raises the following two questions:
\begin{itemize}
  \item How can belief constructions be defined uniformly for classes of partially observable systems that include POMDPs?
  \item How can we prove that such constructions preserve semantics?
\end{itemize}

In this paper, we answer these questions from a coalgebraic perspective.
We model a partially observable system with initial states as 
\emph{pointed partially observable coalgebras}. Concretely, such a coalgebra consists of
a pair of an initial-state map $i\colon I \to S$ and
a coalgebra
$c \coloneqq \langle \delta,\obsv\rangle\colon S \to FTS \times O,
$
where $F\colon \cat{C}\to \cat{C}$ describes the one-step transition type, $T\colon \cat{C}\to \cat{C}$ is a monad describing the branching type, such as nondeterminism and probability, and $O$ is an object of observations.

We then propose a generic construction, called the \emph{coalgebraic belief construction}.
Given a pointed coalgebra $(i, c)$,
the construction produces a new pointed coalgebra
$\Bel(i, c)$, called its \emph{belief coalgebra}.
The idea is to abstract the ordinary belief update for POMDPs: 
one first applies the probabilistic transition to a belief, and then decomposes the resulting distribution into conditional beliefs indexed by observations.
Categorically, the first step is provided by
the coalgebraic determinization~\cite{DBLP:journals/corr/abs-1302-1046},
while the second step is provided by a new structure that we call a \emph{belief decomposition}.

To formulate the second step, we regard state spaces equipped with observations as 
objects of the slice category $\cat{C}/O$,
and
lift the monad $T$ on $\cat{C}$ to a monad 
$\slicemonadO{T}{O}$ on 
the slice category $\cat{C}/O$.
The belief decomposition uses this lifted monad to describe the operation of splitting a $T$-structured state into components indexed by observations, generalizing the decomposition of a belief into conditional beliefs in the  POMDP case.

Our main theorem states that the coalgebraic belief construction is correct:
the semantics of a partially observable coalgebra coincides with that of its belief coalgebra.
In our abstract setting, the belief coalgebra still carries an observation map, whereas the usual belief MDP of a POMDP is fully observable. 
We therefore further provide conditions under which
the semantics of this belief coalgebra agrees with the semantics of its fully
observable counterpart.
This recovers the standard equivalence between a POMDP and its
belief MDP.

We instantiate the framework with the nonempty powerset monad for nondeterminism, the distribution monad for probability, and the semimodule monad for weighted systems.
These instances recover the usual belief construction for POMDPs and yield a new belief construction for partially observable weighted transition systems.
As a corollary, we recover the decidability of termination for partially observable nondeterministic transition systems.

In summary, our contributions are as follows:
\begin{itemize}
  \item a generic coalgebraic framework for belief constructions, based on monads lifted to slice categories and belief decompositions;
  \item a correctness theorem for the coalgebraic belief construction, showing that the semantics of a partially observable coalgebra coincides with that of its belief coalgebra;
  \item sufficient conditions under which partially observable coalgebras can be seen semantically as fully observable ones;
  \item concrete instances for 
  nondeterministic, probabilistic, and weighted transition systems,
  recovering known belief constructions and yielding a new belief construction for partially observable weighted transition systems with the semimodule monad.
\end{itemize}

\noindent
\textbf{Structure}. 
After recalling preliminaries in~\cref{sec:preliminaries}, we lift monads to
slice categories in~\cref{sec:slicingMonads} and define the coalgebraic belief
construction in~\cref{sec:belief_const}. We introduce semantics in~\cref{sec:semBeliefCoal} and prove the correctness theorem in~\cref{sec:correctness}. We then compare partially observable coalgebras and their
fully observable counterparts in~\cref{sec:reachable}, present examples in~\cref{sec:examples}, discuss related work in~\cref{sec:relatedwork}, and
conclude in~\cref{sec:conclusion}.

\section{Preliminaries}
\label{sec:preliminaries}

We recall preliminaries on coalgebras, as well as 
the coalgebraic framework for determinization introduced by Silva et al.~\cite{DBLP:journals/corr/abs-1302-1046}, which has also been applied to trace semantics~\cite{Jacobs0S15,GoyR18,BonchiBCR016,FrankMU22} and graded semantics~\cite{ForsterSWBGM24}.
We refer the reader to \cite{J2016,adamek2025initial} for the background on coalgebras.

Given an endofunctor $F\colon \cat{C} \to \cat{C}$, an \emph{$F$-coalgebra} $c$ is a morphism $c\colon S\rightarrow FS$,
and a \emph{coalgebra morphism} from $c\colon S \to FS$ to $c'\colon S' \to FS'$ is a morphism $f\colon S \to S'$ such that $c' \circ f = Ff \circ c$.
\begin{defi}
  Let $F\colon \cat{C} \to \cat{C}$.
  For an object $I \in \cat{C}$,
  an \emph{($I$-)pointed coalgebra} is a pair $(i,c)$ of a morphism $i\colon I \to S$ and an $F$-coalgebra $c\colon S \to FS$, often written as $(i, c)\colon I \to S \to FS$.
  A \emph{morphism} $f$ between pointed coalgebras $(i,c)\colon I \to S \to FS$ and $(i',c')\colon I \to S' \to FS$ is a morphism $f\colon S \to S'$ such that 
  $f \circ i = i'$ and $c' \circ f = Ff \circ c$.
\end{defi}

\begin{defi}
Let $T \colon \cat{C} \to \cat{C}$ be a monad and let
$F \colon \cat{C} \to \cat{C}$ be a functor.
  A natural transformation $\lambda\colon TF \Rightarrow FT$ is a \emph{distributive law} (of $T$ over $F$) 
  if 
  $F(\eta_X) = \lambda_X \circ \eta_{FX}$ and $F(\mu_X) \circ \lambda_{TX} \circ T(\lambda_X) = \lambda_X \circ \mu_{FX}$ for each $X \in \cat{C}$. 
\end{defi}
A distributive law allows one to combine the branching structure described by $T$ with the one-step transition type described by $F$.
It is known that distributive laws correspond bijectively to liftings of $F$ to the category of Eilenberg--Moore algebras for $T$; see, e.g.,~\cite{J2016}.

\newcommand{\Det}{\mathrm{Det}}
\begin{defi}
Let $T \colon \cat{C} \to \cat{C}$ be a monad,
$F \colon \cat{C} \to \cat{C}$ be a functor,
  and $\lambda\colon TF \Rightarrow FT$ be a distributive law.
  For a coalgebra
$c \colon S \to FTS$,
the \emph{(coalgebraic) determinization} is the $F$-coalgebra given by
$
\Det(c)
\coloneqq
TS\xrightarrow{Tc} TFTS \xrightarrow{\lambda_{TS}} FTTS \xrightarrow{F\mu_S} FTS.
$
\end{defi}

Intuitively, $Tc$ applies the original transition structure pointwise to a
$T$-structured state, $\lambda$ exchanges the transition types $T$ and $F$, and
$F\mu_S$ flattens the resulting double $T$-structure.
This construction generalizes familiar determinization procedures such as the
powerset construction for nondeterministic automata.
\begin{exa}
  Let $T$ be the finite powerset monad $\pfunc\colon \Set \to \Set$ and $F$ be the functor $(\_)^A \times \mathbf{2}\colon \Set \to \Set$ where $A$ is a fixed finite set of actions and $\mathbf{2} = \{0, 1\}$. 
  Then an $FT$-coalgebra $\langle \delta, \mathrm{acc} \rangle\colon S \to (\pfunc S)^A \times \mathbf{2}$ is precisely a nondeterministic automaton: 
  for each state $s$ and action $a$,
  $\delta(s)(a)$ is the set of successors of $s$ under $a$, and $\mathrm{acc}(s)$ indicates whether $s$ is an accepting state.
  We define a distributive law $\lambda\colon \pfunc ((\_)^A \times \mathbf{2}) \Rightarrow (\pfunc(\_))^A \times \mathbf{2}$ by 
$
\lambda_X(U) \coloneqq \big\langle a \mapsto \bigcup_{(f, t) \in U} f(a), \bigvee_{(f, t) \in U} t \big\rangle,
$
for each $X \in \Set$ and $U \in \pfunc(X^A \times \mathbf{2})$.
Then $\Det(\langle \delta, \mathrm{acc} \rangle)$ is the standard powerset construction for nondeterministic automata. Its state space is $\pfunc S$, the successor of a state $U$ under an action $a$ is the state $\bigcup_{s \in U} \delta(s)(a)$, and $U$ is accepting if $\mathrm{acc}(s) = 1$ for some $s \in U$.
\end{exa}

\section{Slicing Monads}
\label{sec:slicingMonads}

The coalgebraic determinization recalled in~\Cref{sec:preliminaries} applies to coalgebras of the form $c\colon S \to FTS$. 
For the belief construction, however,
states are additionally equipped with observations.
To incorporate observations categorically, we work in slice categories. 

For an object $O \in \cat{C}$, recall that the slice category 
$\cat{C}/O$ has as objects morphisms $f\colon X \to O$,
and as morphisms from $f\colon X \to O$ to $g\colon Y \to O$ 
morphisms $h\colon X \to Y$ satisfying $g \circ h = f$.
Thus an object of $\cat{C}/O$ may be regarded as a state space $X$ equipped with an observation map into $O$.
By abuse of notation, for a morphism $f$ in a slice category, we sometimes write $f$ for its underlying morphism in $\cat{C}$.

In this section, we aim to lift the monad $T$ on $\cat{C}$ to a monad on each slice category $\cat{C}/O$. We then study how these lifted monads behave under change of the observation
object and how they relate to the original monad $T$ via the forgetful functor
$\dom_O \colon \cat{C}/O \to \cat{C}$. 

We begin by recalling oplax monad morphisms. 
\begin{defi}
Let \(T = (T,\eta^T,\mu^T)\) be a monad on \(\cat{C}\), and \(S = (S,\eta^S,\mu^S)\) a monad on \(\cat{D}\).
An \emph{(oplax) monad morphism} from \((\cat{C},T)\) to \((\cat{D},S)\) 
is a tuple $(F, \psi)$ of a functor \(F \colon \cat{C} \to \cat{D}\) and 
a natural transformation $\psi \colon F T \Rightarrow S F$
such that 
  $\eta_{FX}^S = \psi_X \circ F(\eta_{X}^T)$ and $\mu_{FX}^S \circ S\psi_X \circ \psi_{TX} = \psi_X \circ F(\mu_{X}^T)$ for each $X \in \cat{C}$. 
\end{defi}

\begin{lem}[{Cf.~\cite[Sec.~2]{Hernandez2014ApplicationsOT},\cite{10.1007/BFb0063105}}] \label{lem:adj_monad}
  Consider
functors $F\colon \cat{C}' \to \cat{C}, G\colon \cat{D}' \to \cat{D}$ and adjunctions $L \dashv R\colon \cat{D} \to \cat{C}$ and $L' \dashv R'\colon \cat{D}' \to \cat{C}'$. Assume that $LF = GL'$, and
let $\alpha\colon FR' \Rightarrow RG$ be the mate of $LF = GL'$.
Then $(F, \alpha L')$
is 
a monad morphism from $(\cat{C}', R'L')$ to $(\cat{C}, RL)$.
\end{lem}
We use \Cref{lem:adj_monad} twice in this section:
to obtain
the base-change morphisms \((\Sigma_u,\theta_u)\) in \Cref{prop:ol_T}
and the forgetful morphism \((d,\iota)\) in \Cref{prop:iota}.

For the remainder of this section,
we assume that 
$\cat{C}$ is a category with pullbacks. 
Let $T$ be a monad on $\cat{C}$, and
for each $f\colon X \to O$ in $\cat{C}$,
we fix a pullback of $Tf$ along $\eta^T_O$.
\newcommand{\domofslicenew}[3]{B_{#3}}
\begin{defi}
  For each object \(O\in \cat C\), we define an endofunctor
  $\slicemonadO{T}{O} \colon \cat C/O \to \cat C/O$.
  For an object \(f:X\to O\), the object \(\slicemonadO{T}{O}(f)\colon \domofslicenew{T}{O}{f} \to O\) is defined by the pullback of $Tf$ along $\eta_O$, as in the left diagram below.
  For a morphism $h\colon f \to g$ in $\cat{C}/O$, we define $\slicemonadO{T}{O}(h)$ to be the unique morphism induced by universality of the pullback for $\domofslicenew{T}{O}{g}$, as in the right diagram below.
  \begin{equation} \label{diag:pullback_eta_o}
    \xymatrix{
      \domofslicenew{T}{O}{f} \pullbackmark{0, 1}{1, 0} \ar[r]^{\slicemonadO{T}{O}f} \ar[d]_{\iota_f} &O \ar[d]^{\eta_O} \\
      TX \ar[r]^{Tf} &TO
    } \qquad
    \xymatrix{
      \domofslicenew{T}{O}{f} \ar[d]_{\iota_f} \ar@/^1em/[rr]^{\slicemonadO{T}{O}f} \ar@{.>}[r]_{\slicemonadO{T}{O}h}
      &\domofslicenew{T}{O}{g} \pullbackmark{0, 1}{1, 0} \ar[r]_{\slicemonadO{T}{O}g} \ar[d]_{\iota_g} &O \ar[d]^{\eta_O} \\
      TX \ar[r]^{Th} \ar@/_1em/[rr]_{Tf} &TY \ar[r]^{Tg} &TO
    }
  \end{equation}
\end{defi}

\begin{prop} \label{prop:ol_T}
  For each object $O \in \cat{C}$,
the endofunctor \(\slicemonadO{T}{O}\) carries a canonical monad structure.
Moreover, for each morphism \(u:O\to O'\), the post-composition functor
$\Sigma_u \coloneqq u\circ(\_):\cat C/O \to \cat C/O'$
extends to a monad morphism
$(\Sigma_u,\theta_u):(\cat C/O,\slicemonadO{T}{O}) \to (\cat C/O',\slicemonadO{T}{O'})$,
where \(\theta_u:\Sigma_u\slicemonadO{T}{O} \Rightarrow \slicemonadO{T}{O'}\Sigma_u\) is the natural transformation whose component at \(f:X\to O\) is induced
by universality of the pullback of $\eta_O'$ along $T(u \circ f)$.
\end{prop}
We refer to
\cref{prop:appol_T} 
for the proof.

\newcommand{\Mnd}{\mathbf{Mnd}}
\begin{rem} \label{rem:indexed_ol_T}
  The result in \Cref{prop:ol_T} gives the following 
  $2$-categorical reformulation.
Define a functor $\Sl\colon \cat{C} \to \Cat$ by
$\mathrm{Sl}(O)\coloneqq \cat{C}/O$ and $\mathrm{Sl}(u)\coloneqq \Sigma_u$.
Then the assignment
$O \in \cat{C} \mapsto (\cat{C}/O,\slicemonadO{T}{O})$ 
with $(u\colon O \to O') \mapsto (\Sigma_u, \theta_u)$
defines
a functor
$\ol{T}\colon \cat{C} \to \Mnd$
whose composite with the forgetful $2$-functor $\Mnd \to \Cat$ is $\Sl$.
Here $\Mnd$ is the 2-category of monads and oplax monad morphisms.
\end{rem}
We sometimes omit the subscript $O$ in $\slicemonadO{T}{O}$ when it is clear from the context.

The forgetful functor $\dom\colon \cat{C}/O \to \cat{C}$
is also compatible with the monads $\slicemonad{T}{O}$ and $T$ as below,
see 
\Cref{ap:proof_iota} 
for the proof.
Please note that $\domofslice{T}{O}{f} = \domofslicenew{T}{O}{f}$ in \eqref{diag:pullback_eta_o} for each $f \in \cat{C}/O$. 
\begin{prop} \label{prop:iota}
  There is a monad morphism $(\dom, \iota)\colon (\cat{C}/O, \slicemonad{T}{O}) \to (\cat{C}, T)$
  given by the morphisms $(\iota_f)_{f \in \cat{C}/O}$ defined in the pullback square on the left in \eqref{diag:pullback_eta_o}.
\end{prop}

\section{Coalgebraic Belief Construction} \label{sec:belief_const}

We introduce our
coalgebraic belief construction, which transforms a coalgebra with an observation map into a coalgebra, called its belief coalgebra.
The construction is based on 
the coalgebraic
determinization (see~\Cref{sec:preliminaries})
together with 
the monads $\slicemonadO{T}{O}$ on slice categories
introduced in \Cref{sec:slicingMonads}.

We begin by introducing coalgebras equipped with observations.
Here, we adopt a wide subcategory $\cat{O}$ of $\cat{C}$, that is, a subcategory containing all objects of $\cat{C}$, to specify admissible observation morphisms.

\begin{defi}
  Let $\cat{C}$ be a cartesian category,
  $\cat{O}$ be a wide subcategory of $\cat{C}$,
  and $G$ be an endofunctor on $\cat{C}$.
  A \emph{partially observable $G$-coalgebra} (or \emph{PO coalgebra} for short) is a morphism 
  $\langle \delta, \obsv\rangle\colon S \to GS \times O$,
  and a morphism from $\langle \delta, \obsv\rangle\colon S \to GS \times O$ to $\langle \delta', \obsv'\rangle\colon S' \to GS' \times O'$
  is a pair of morphisms $f\colon S \to S'$ in $\cat{C}$ and $g\colon O \to O'$ in $\cat{O}$ such that $f\colon \delta \to \delta'$ is a $G$-coalgebra morphism and $g \circ \obsv = \obsv' \circ f$.

  For an object $I \in \cat{C}$, 
  an \emph{($I$-)pointed partially observable $G$-coalgebra} (or \emph{pointed PO coalgebra}) is 
  a pair $(i, c)$ of a morphism $i\colon I \to S$ and a PO coalgebra $c\colon S \to GS \times O$, often written as $(i, c)\colon I \to S \to GS \times O$.
  A morphism $(f, g)\colon (i, c) \to (i', c')$ between pointed PO coalgebras 
  is a morphism $(f, g)\colon c \to c'$ between PO coalgebras  such that 
  $f \circ i = i'$.
We write $\CoalgOcov{G}{}$ for the category of partially observable $G$-coalgebras and $\CoalgOcov{G}{I}$ for the category of $I$-pointed partially observable $G$-coalgebras. 
\end{defi}
We say that $(i, \langle \delta, \obsv\rangle)\colon I \to S \to GS \times O$ is \emph{fully observable} if $S = O$ and $\obsv = \id_S$.

\begin{asm}
  \label{asm:minimal}
  We fix the following data:
\begin{itemize}
  \item $\cat{C}$ is a category with finite products and pullbacks,
  \item $\cat{O}$ is a wide subcategory of $\cat{C}$, 
  \item 
  $F \colon \cat{C} \to \cat{C}$ is a
functor and
$T\colon \cat{C} \to \cat{C}$ is a monad,
  \item $\lambda \colon TF \Rightarrow FT$ is a distributive law,
  \item $I$ is an object of $\cat{C}$.
\end{itemize}
\end{asm}

Throughout the section, 
we work
under \Cref{asm:minimal} and restrict our attention to $I$-pointed partially observable $FT$-coalgebras $(i, \langle \delta, \obsv\rangle)\colon I \to S \to FT S \times O$.

  For instance, 
our framework specializes to
POMDPs 
by taking 
$\cat{C} = \Set$,
$F = (\_ \times \mathbb{R}_{\geq 0})^A$, and $T = \mathcal{D}$
where $A$ is the set of actions and $\mathcal{D}$ is the finite distribution monad.
Although our framework assumes deterministic observations $S \to O$,
standard POMDPs with stochastic observations can be encoded by expanding the state space; see \cite[Remark 1]{DBLP:journals/ai/ChatterjeeCGK16}.

Our aim is to associate with each pointed PO coalgebra a pointed $FT$-coalgebra, called the belief coalgebra, whose state space is $\domofslice{T}{O}{\obsv}$, obtained by
the pullback of $\eta_O$ along $T(\obsv)$ as in~\eqref{diag:pullback_eta_o}.
The following is a key ingredient for the construction of the belief coalgebra.

\newcommand{\flatten}{\mathsf{flat}}
\begin{defi}[belief decomposition] \label{def:alpha}
For each object $O\in\cat{O}$, 
define 
$\flatten^O$ to be the natural transformation $\mu\, \dom \circ T\iota\colon T \dom \slicemonadO{T}{O} \Rightarrow T \dom$ between functors of the type $\cat{C}/O \to \cat{C}$.
A \emph{belief decomposition on $T$} consists of a family of natural
transformations
$\alpha^O\colon T \dom \Rightarrow T \dom \slicemonadO{T}{O}$
for each $O\in\cat{O}$ such that 
(i) each $\alpha^O$ is a section of $\flatten^O$, i.e.~$\flatten^O \circ \alpha^O = \id_{T \dom}$,
and (ii) for each $u\colon O\to O'$ in $\cat{O}$,
the following equality of natural transformations from $\cat{C}/O$ to $\cat{C}$ holds:
$T \dom_{O'}\theta_u \circ \alpha^{O} = \alpha^{O'} \Sigma_u\colon T\dom_O \Rightarrow \dom_{O'} \slicemonadO{T}{O'}\Sigma_u$, where we use $\dom_{O'} \Sigma_u  = \dom_O$.
\end{defi}

Allowing observation objects to vary reflects the functoriality
of the slice-category construction in the observation object.
Accordingly,
Condition~(ii) requires
the family of belief decompositions to be coherent 
under admissible changes of observation objects,
see \Cref{ap:indexed-belief-decomp} for details.
In all examples, we take $\cat{O}$ to be the wide subcategory of $\cat{C}$ consisting of all monomorphisms.

  When no confusion arises, we omit the superscripts of $\alpha^O$ and $\flatten^O$.
  We now fix a belief decomposition $\alpha$, which in turn induces the following coalgebraic belief construction.
\begin{defi}[coalgebraic belief construction] \label{def:bel}
Let
$(i, \langle \delta,\obsv\rangle) \colon I \to S \to FTS \times O$
be a pointed PO coalgebra,
and write $c$ for $\langle \delta,\obsv\rangle$.
We define its \emph{belief coalgebra} by
\[
\Bel(i, \langle \delta,\obsv\rangle) \coloneqq \big(\eta^{\slicemonad{T}{O}}_{\obsv} \circ i, \langle c^\Bel, \slicemonad{T}{O}(\obsv)\rangle\big) \colon I  \to \domofslice{T}{O}{\obsv} \to FT(\domofslice{T}{O}{\obsv}) \times O
\]
where 
the transition part $c^\Bel$ is defined by the following composite:
\begin{align*}
c^\Bel
&\coloneqq \domofslice{T}{O}{\obsv} \xrightarrow{\iota_{\obsv}} TS \xrightarrow{\Det(\delta)} FTS \xrightarrow{F\alpha_{\obsv}} FT(\domofslice{T}{O}{\obsv}).
\end{align*}
\end{defi}
  Here, following the convention introduced at the beginning of \Cref{sec:slicingMonads},
  $\eta^{\slicemonad{T}{O}}_{\obsv}$ denotes the underlying morphism in $\cat{C}$ of the unit morphism  in $\cat{C}/O$.

The belief coalgebra is obtained by combining the determinization
$\Det(\delta)\colon TS \to FTS$
with 
the inclusion morphism
$\iota_\obsv\colon \domofslice{T}{O}{\obsv} \to TS$ and
the belief decomposition
$\alpha_{\obsv} \colon TS \to T(\domofslice{T}{O}{\obsv})$.
More precisely, the morphism $\iota_{\obsv} \colon \domofslice{T}{O}{\obsv} \to TS$
first forgets the observation-indexed structure, then $\Det(\delta)$ applies the
usual coalgebraic determinization, and finally $\alpha_{\obsv}$ reorganizes the
resulting $T$-structured state according to the observation map $\obsv$.

The belief construction defined above extends to a functor on the category of pointed PO coalgebras, as follows. See 
\Cref{ap:proof_bel} 
for the proof.
\begin{prop} \label{prop:bel}
The assignment $(i, c) \mapsto \Bel(i, c)$ extends to a
functor \[\Bel \colon \CoalgOcov{FT}{I} \to \CoalgOcov{FT}{I}\]
by mapping
a morphism
$(f, g) \colon (i, \langle \delta,\obsv\rangle) \to (i', \langle \delta',\obsv'\rangle)$
to
$\left(
\dom\slicemonad{T}{O'}(f) \circ \dom(\theta_g)_{\obsv},
\, g
\right)$,
where the first component is as below:
\[
\domofslice{T}{O}{\obsv} = \dom \Sigma_g \slicemonad{T}{O}(\obsv) \xrightarrow{\dom (\theta_g)_\obsv} \dom \slicemonad{T}{O'}\Sigma_g(\obsv) \xrightarrow{\dom\slicemonad{T}{O'}(f)} \domofslice{T}{O'}{\obsv'}.\]
\end{prop}

\begin{exa}[nondeterminism]
  \label{ex:nondet}
We use partially observable nondeterministic transition systems as a running example throughout the paper.
 Let $T$ be the  nonempty finite powerset monad $\pefunc$ on $\Set$, and $F \coloneqq (\_) + \{\checkmark\}$, where $\checkmark \, (\not \in S)$ denotes the designated terminal state. 
We define a \emph{distributive law} 
$\lambda\colon \pefunc((\_) + \{\checkmark\}) \rightarrow \pefunc(\_) + \{\checkmark\}$ by, for each $X \in \Set$,
$\lambda_X (U) \coloneqq \checkmark$ if $U = \{\checkmark\}$, and $\lambda_X (U) \coloneqq U\cap X$ otherwise. 
We then define a \emph{belief decomposition} $\{\alpha^O_f\colon (\pefunc \dom)(f) \rightarrow (\pefunc\dom\ol{\pefunc})(f)\}_{f \colon X \to O}$ by
 \[
   \alpha^O_f(U) \coloneqq \big\{ f^{-1}(o) \cap U \, \big| \,  o\in O\text{ s.t. }f^{-1}(o)\cap U\not = \emptyset\big \}, \text{ for any $U\in \pefunc(X)$.}
 \]
See
\cref{ap:proofAlphaNondet} 
for the proof that $\alpha$ is indeed a belief decomposition.
 Given a pointed PO coalgebra $(i, \langle \delta,\obsv\rangle) \colon I \to S \to \big(\pefunc(S) + \{\checkmark\}\big) \times O$, 
 its \emph{belief coalgebra} $\Bel(i, \langle \delta,\obsv\rangle)\colon I \rightarrow  \domofslice{\pefunc}{O}{\obsv} \to \big(\pefunc(\domofslice{\pefunc}{O}{\obsv}) + \{\checkmark\}\big) \times O$ is induced by the distributive law and the belief decomposition, where 
 the state object $\domofslice{\pefunc}{O}{\obsv}$ is the set of $U \in \pefunc(S)$ such that there exists $o \in O$ with $U \subseteq \obsv^{-1}(o)$, and
 the transition part $c^\Bel$ is given by 
 \begin{align*}
  c^\Bel(U) = \begin{cases}
    \big\{ V_U \cap \obsv^{-1}(o)  \, \big|\,  o\in O\big\} \setminus \{\emptyset\} &\text{ if }V_U\not = \emptyset,\\
    \checkmark &\text{ otherwise, }
  \end{cases}
 \end{align*}
where $V_U \coloneqq \bigcup_{s \in U \text{ s.t.}~\delta(s) \neq \checkmark}\delta(s)$.

Intuitively, a state of its belief coalgebra is a set of states of the original coalgebra that share the same observation.
Its nondeterministic transition is defined over partitions of all possible transitions, where the partitions are induced by observations.
\end{exa}

\begin{rem}[the design choice for nondeterminism] \label{rem:design_choice}
  You may wonder whether we can define $T$ and $F$ by the powerset monad $T\coloneqq \pfunc$ and $F\coloneqq \id$, respectively. 
  This design choice is not suitable for us since it does not satisfy an assumption we will make in~\Cref{sec:correctness}
  for the correctness of the belief construction.
\end{rem}

We conclude this section with two basic properties of the belief construction.
We first consider the fully observable case, and show that the belief construction
does not produce new states up to canonical isomorphism.
We then make precise how the belief construction is related to ordinary coalgebraic
determinization.
See 
\Cref{ap:proof_bel_id} 
for the proof.
\begin{prop} \label{prop:bel_delta_id}
For every $FT$-coalgebra $\delta \colon S \to FTS$, one has
$\Bel(\id_S, \langle \delta, \id_S\rangle) \cong (\id_S, \langle \delta, \id_S\rangle)\colon S \to S \to FTS \times S$.
\end{prop}

\begin{prop} \label{prop:bel_flat}
For each pointed PO coalgebra 
$(i, \langle \delta,\obsv\rangle) \colon I \to S \to FTS \times O$, where $c = \langle \delta,\obsv\rangle$,
one has
$F(\flatten_\obsv) \circ c^\Bel
=
\Det(\delta) \circ \iota_{\obsv}$.
\end{prop}
\begin{proof}
This is immediate from the first condition in \Cref{def:alpha}.
\end{proof}
Thus, after collapsing the observation-wise decomposition via $F(\flatten_\obsv)$, 
the belief transition agrees with the ordinary
determinization along the morphism
$\iota_{\obsv}$.

\section{Semantics of Pointed Partially Observable Coalgebras}
\label{sec:semBeliefCoal}
One key property of the belief construction is that it preserves semantics,
such as the maximal expected reward in the case of POMDPs.
To make this statement precise, 
we first introduce a scheduler-based semantics of pointed PO coalgebras:
for each scheduler, one obtains a semantics of the coalgebra under the scheduler,
and the overall semantics is obtained by taking the join over all schedulers.

Accordingly, 
we restrict attention to 
partially observable $FT(\_)^A$-coalgebras 
\[(i, \langle \delta,\obsv\rangle)\colon I \to S \to (FTS)^A \times O,\]
where $A \in \cat{C}$ is a fixed object of actions.
Intuitively, a scheduler specifies a choice of action based on a history of observations, as we make precise in
the next subsections.

Before defining the semantics of pointed partially observable $FT(\_)^A$-coalgebras, 
we prepare several auxiliary notions.

\begin{defi}[ordered object~\cite{DBLP:journals/mscs/AguirreKK22}]
An object \(\Omega \in \cat{C}\) is an \emph{ordered object} if, for each object
\(X \in \cat{C}\), the hom-set \(\cat{C}(X,\Omega)\) is a complete lattice and
each precomposition map preserves arbitrary joins.
For an endofunctor \(F \colon \cat{C} \to \cat{C}\) and an ordered object \(\Omega\), an \(F\)-algebra
$\tau \colon F\Omega \to \Omega$
is \emph{monotone} if, for each object \(X \in \cat{C}\), the map $\tau \circ F(\_)\colon \cat{C}(X,\Omega) \to \cat{C}(FX,\Omega)$ is monotone.
\end{defi}
We write $\bot_X$ and $\leq_X$ (or simply $\bot$ and $\leq$) for the bottom element and the order on $\cat{C}(X,\Omega)$, respectively.

\begin{asm}
In the rest of the paper,
we assume~\cref{asm:minimal} and the following.
  \label{asm:semBelief}
  \begin{itemize}
    \item The category $\cat{C}$ is a cartesian closed category with countable coproducts.
    \item The functor $F$ and the monad $T$ are strong.
    We write their strengths as:
$\stl^F_{X,Y}\colon X \times FY \to F(X \times Y)$ and
$\stl^T_{X,Y}\colon X \times TY \to T(X \times Y)$ for each $X, Y \in \cat{C}$.
  \item The distributive law $\lambda$ satisfies $\lambda_{X \times Y} \circ T(\stl^F) \circ  \str^T = F\stl^T \circ \stl^F \circ X \times \lambda_{Y}$ for each $X, Y \in \cat{C}$.
  \item We fix an object $A \in \cat{C}$.
  \item We fix an ordered object $\Omega \in \cat{C}$ and a monotone algebra $\tau\colon FT\Omega \to \Omega$.
  \end{itemize}
\end{asm}
\newcommand{\nil}{\mathsf{nil}}
\newcommand{\cons}{\mathsf{cons}}

\noindent\textbf{Notation.}
We write
$(\_)^*$ and $(\_)^+$ for the endofunctors 
$\coprod_{n \in \mathbb{N}} (\_)^n$ and $\coprod_{n \in \mathbb{N}_{>0}} (\_)^n$ on $\cat{C}$, respectively,
and write 
$\nil$ for the canonical natural transformation $\id \Rightarrow (\_)^+$,
and $\cons_O\colon O \times O^* \to O^+$ for the canonical isomorphism for each $O \in \cat{C}$.
We write $f^\dagger$ for the adjoint transpose of a morphism $f$ under the relevant adjunction $(\_) \times X \dashv (\_)^X$ (for $X \in \cat{C}$). 
We write $\ev$ for the components of its counit, i.e.~evaluation morphisms, omitting subscripts when they are clear from the context.
For each product $X_1 \times \cdots X_n$, the $i$-th projection ($i \in \{1, \cdots, n\}$) is denoted by $\pi_i$.

\medskip

Under \Cref{asm:semBelief}, 
we obtain a strength $\str^{FT}$ of $FT$ given compositionally from the strength of $T$ and $F$.
Moreover,
the distributive law of $T$ over $F$ induces a distributive law of $T$ over $F(\_)^A$ by the following lemma.
Hence the construction of \Cref{sec:belief_const} applies to 
$(i, \langle \delta,\obsv\rangle) \colon I \to S \to (FTS)^A \times O$ simply by replacing $F$ with $F(\_)^A$.
\begin{lem}
  Let $F\colon \cat{C} \to \cat{C}$, $T$ be a strong monad on $\cat{C}$, and $\lambda\colon TF \Rightarrow FT$ be
  a  distributive law. 
  Then, there is a distributive law $\lambda'\colon T\circ (F(\_))^A \Rightarrow (F(\_))^A\circ T$ defined by 
  \[
    \lambda'_X \coloneqq  T((FX)^A) \xrightarrow{\stlift_{FX, A}} (TFX)^A \xrightarrow{(\lambda_X)^A} (FTX)^A,   
  \]
  where $\stlift_{FX, A}$ is the canonical map induced by the strength $\stl^T$. 
\end{lem}

\begin{definition}
For an $FT$-coalgebra $c\colon S \to FTS$, define a monotone map
$\Phi_{c} \colon \cat{C}(S,\Omega) \to \cat{C}(S,\Omega)$
by
$\Phi_{c}(f) := \tau \circ FTf \circ c$.
Then the semantics of an $I$-pointed $FT$-coalgebra $(i, c)\colon I \to S \to FTS$ is given as
$\bigvee_{n \in \nat} \Phi^n_c(\bot_S) \circ i$.
\end{definition}

\begin{prop}\label{prop:Phi-coalg-morphism}
If $f\colon c \to c'$ is an $FT$-coalgebra morphism, then
$\Phi_c^n(\bot_S) = \Phi_{c'}^n(\bot_{S'}) \circ f$
for each $n \in \mathbb{N}$.
Consequently, for each 
morphism $f\colon (i, c) \to (i', c')$ between $I$-pointed coalgebras,
the semantics of $(i, c)$ is equal to that of $(i', c')$, that is,
\[
\bigvee_{n \in \nat} \Phi^n_c(\bot_S) \circ i
= \bigvee_{n \in \nat} \Phi^n_{c'}(\bot_{S'}) \circ i'.
\]
\end{prop}

\begin{proof}
The first statement can be easily proved by induction on $n$.
The second statement follows immediately from the first.
\end{proof}

\begin{exa}[nondeterminism]
  \label{ex:nondet2}
  We continue the example shown in Ex.~\ref{ex:nondet}. 
  We define the ordered object $\Omega$ to be $\bset$, where $\bset \coloneqq\{\btrue, \bfalse\}$ is the Boolean domain with the standard total order $\bfalse <\btrue$.
  The order on $\Set(X,\bset)$ is given by the pointwise order.  
  We define a monotone algebra $\tau\colon \pefunc(\bset) + \{\checkmark\}\rightarrow \bset$ by 
    $\tau(S) \coloneqq \land S\text{ for $S\in \pefunc(\bset)$, and } \tau(\checkmark) = \btrue$. 
  Given a coalgebra $c\colon S \to \pefunc(S) + \{\checkmark\}$, the monotone map $\Phi_{c} \colon \cat{C}(S,\bset) \to \cat{C}(S,\bset)$
  is 
  \begin{align*}
    \Phi_{c}(f)(s) = \begin{cases*}
        \quad \bfalse \quad \text{ if there exists $s' \in c(s) \cap S$ s.t.~$f(s') = \bfalse$},\\
        \quad \btrue \quad \text{ otherwise.}
    \end{cases*} 
  \end{align*}
  For an $I$-pointed coalgebra $(i,c)\colon I \to S \to \pefunc(S) + \{\checkmark\}$,
the semantics
$\bigvee_{n \in \nat} \Phi_c^n(\bot_S) \circ i \colon I \to \bset$
is then given by
$\bigl(\bigvee_{n \in \nat} \Phi_c^n(\bot_S) \circ i\bigr)(x)=\btrue$
if and only if all paths from $i(x)$ in $c$ terminate
(equivalently, reach the designated state $\checkmark$).
\end{exa}

\newcommand{\upd}{\mathrm{upd}}
\newcommand{\crt}{\mathrm{crt}}

We then introduce two scheduler-based semantics for pointed PO coalgebras. 
In \Cref{subsec:schSemantics}, we define a semantics
by carrying schedulers in the state space, and in \Cref{subsec:histSemantics},
we give an equivalent formulation based on state histories.

\subsection{Semantics based on Coalgebras Carrying Schedulers} \label{subsec:schSemantics}
For an object of observations $O$,
a \emph{scheduler} is a morphism $u\colon O^+ \to A$ in $\cat{C}$, assigning an action to each nonempty history of observations.
One may also present schedulers as morphisms
$O \times (A\times O)^* \to A$.
In the deterministic setting, however, these two presentations are equivalent,
because the past actions are recursively determined by the past observations.

As the history grows, the
scheduler is consumed accordingly, and the remaining choice mechanism can be carried along
as part of the state. This yields an \(FT\)-coalgebra on \(S \times A^{O^\ast}\).

\begin{defi} \label{def:pol}
For a PO coalgebra
$c = \langle \delta,\obsv\rangle\colon S \to (FTS)^A \times O$,
we define an $FT$-coalgebra
$\Pol(c)\colon S \times A^{O^*} \to FT(S \times A^{O^*})$
as the composite
\[
S \times A^{O^*}
\cong S \times A \times A^{O^+} \!
\xrightarrow{\delta^\dagger \times \id} \!
FTS \times A^{O^+} \!
\xrightarrow{\str^{FT}} \!
FT(S \times A^{O^+}) \!
\xrightarrow{FT\langle \pi_1, \ev \circ (\obsv \times \id)\rangle}\!
FT(S \times A^{O^*}).
\]

Let $\CoalgOctr{(FT(\_))^A}$ be the category of partially observable $(FT(\_))^A$-coalgebras and morphisms that are contravariant in the observation part; that is, a morphism from $\langle \delta,\obsv\rangle \colon S \to (FTS)^A \times O$ to $\langle \delta',\obsv'\rangle \colon S' \to (FTS')^A \times O'$ is a pair of morphisms $f\colon S \to S'$ and $g\colon O' \to O$ such that  $\delta' \circ f = (FTf)^A \circ \delta$ and $\obsv = g \circ \obsv' \circ f$.

The construction defined in \Cref{def:pol} extends to a functor
\[
\Pol\colon \CoalgOctr{(FT)^A}
\to
\Coalg{FT}
\]
with $\Pol(f, g) \coloneqq f \times A^{g^*}$
where 
$A^{g^*}\colon A^{O^*} \to A^{O'^*}$ is the morphism induced by $g$.
\end{defi}

An objective $\Vc{i, c}$ of an $I$-pointed $FT(\_)^A$-coalgebra $(i, c)$ is defined as
the join of the semantics of the $I$-pointed $FT$-coalgebra $(\langle \id, (u \circ \cons_O)^\dagger \circ \obsv\rangle \circ i, \Pol(c))\colon I \to S \times A^{O^*} \to FT(S \times A^{O^*})$ for all $u \colon O^+ \to A$:
\[
\Vc{i, c}
\coloneqq
\bigvee_{u\colon O^+ \to A,\; n\in\mathbb{N}}
\Phi_{\Pol(c)}^n(\bot)\circ \langle \id, (u \circ \cons_O)^\dagger \circ \obsv\rangle \circ i
\;\colon\;
I \to \Omega.
\]

\begin{exa}[nondeterminism]
  \label{ex:nondetOb}
  We continue Ex.~\ref{ex:nondet2}. 
  Consider a coalgebra $c = \langle \delta,\obsv\rangle\colon S \to (\pefunc(S) + \{\checkmark\})^A \times O$
  and a scheduler $u\colon O^{+}\rightarrow A$.
  For each $s \in S$, let $\tilde u_s\colon O^* \to A$ be $\lambda \vec{o}.u(\obsv(s) \vec{o})$
  where $\obsv(s) \vec{o}$ is the sequence given by prepending $\obsv(s)$ to $\vec{o}$.
  Then the composite
  $\Phi_{\Pol(c)}^n(\bot)\langle \id, (u \circ \cons)^\dagger \circ \obsv \rangle\colon S\rightarrow \bset$ maps $s \in S$ to
  \begin{align*}
    \Phi_{\Pol(c)}^n(\bot)(s, \tilde u_s) = \begin{cases}
      \btrue &\text{if under $u$, all paths from $s$ reach $\checkmark$ in $c$ within $n$ steps,}\\
      \bfalse &\text{otherwise.}
    \end{cases}
  \end{align*}
  Consequently, given an initial element $x\in I$, the objective $\Vc{i, c}$ is given by 
  \begin{align*}
    \Vc{i, c}(x) =  \begin{cases*}
      \btrue &\text{if under some $u\colon O^{+}\rightarrow A$, all paths from $i(x)$ eventually reach $\checkmark$ in $c$,}\\
      \bfalse &\text{otherwise.}   
    \end{cases*}
  \end{align*}
\end{exa}

\newcommand{\ext}{\mathsf{ext}}
\subsection{Semantics based on State-History Coalgebras} \label{subsec:histSemantics}

We now give an equivalent presentation of the objective $\Vc{i, c}$ in which
the state-history is stored explicitly in the state space.

\begin{defi} \label{def:pospol}
For a PO coalgebra 
$c = \langle \delta,\obsv\rangle \colon S \to (FTS)^A \times O$
and a morphism $u\colon O \to A$,
we define an $FT$-coalgebra
$c_u\colon S \to FT(S)$
to be
$\ev \circ \langle \delta, u \circ \obsv\rangle$.
\end{defi}

Let
$\ext_S$ be the canonical concatenation morphism $S^+ \times S\to S^+$ 
(or $S^+ \times S^* \to S^+$),
and
$\last\colon S^+ \to S$ be the morphism sending a sequence to its last element.

\begin{defi}\label{def:history-coalg}
Let 
$(i, \langle \delta, \obsv\rangle) \colon I \to S \to (FTS)^A \times O$ be a pointed PO coalgebra, and $c$ be the second part $\langle \delta,\obsv\rangle$.
We define 
a pointed PO coalgebra 
$(\nil \circ i, 
\Hist(c))\colon I \to S^+ \to \big(FT(S^+)\big)^A \times O^+$
where $\Hist(c) \coloneqq \langle f, \obsv^+\rangle$ and $f$ is the composite
\[
S^+
\xrightarrow{\langle \id,\,
\delta \circ \last\rangle}
S^+ \times (FTS)^A
\xrightarrow{(\stl \circ (\id_{S^+} \times \ev))^\dagger}
\big(FT(S^+ \times S)\big)^A
\xrightarrow{FT(\ext)^A}
\big(FT(S^+)\big)^A.
\]
\end{defi}
\newcommand{\Hc}[1]{H_{#1}^*}
An objective $\Hc{i, c}$ of $c$ can then be defined as
the join of the semantics of the $I$-pointed $FT$-coalgebra $(\nil \circ i, \Hist(c)_u)\colon I \to S^+ \to FT(S^+)$ for all $u \colon O^+ \to A$:
\[
\Hc{i, c} \coloneqq
\bigvee_{u \colon O^+ \to A,\; n \in \mathbb{N}}
\Phi_{\Hist(c)_u}^{n}(\bot) \circ \nil_S \circ i,
\]
which is equal to $\Vc{i, c}$ as shown below. See 
\Cref{ap:proof_history_semantics} 
for the proof.

\begin{prop}\label{prop:history-semantics}
  The equation
$\Vc{i, c}
= \Hc{i, c}$ holds.
\end{prop}

\newcommand{\cmp}{\mathrm{cmp}}

\section{Correctness of the Belief Construction} \label{sec:correctness}

We now turn to the main theorem of this paper: the correctness of the belief construction. We
show that the belief construction preserves the semantics of pointed PO coalgebras, in the sense that 
for each pointed PO coalgebra $(i, c)\colon I \to S \to (FTS)^A \times O$, 
the equation $\Vc{i, c}=\Vc{\Bel(i, c)}$ holds.

Throughout the section, we fix a pointed PO coalgebra $(i, \langle \delta, \obsv\rangle)\colon I \to S \to (FTS)^A \times O$ and write $c$ for $\langle \delta, \obsv\rangle$.

For a coalgebra $c\colon S \to GS$ of an endofunctor $G\colon \cat{C} \to \cat{C}$ and $n \in \mathbb{N}$,
we write
$c^n \colon S \to G^n S$
for the morphism defined inductively by $c^0 = \id_S$ and
$c^{n+1} := G(c^n)\circ c$.
We use the same notation for algebras $\tau\colon G\Omega\rightarrow \Omega$.
For each \(n \ge 1\), let
$\lambda^n \colon (FT)^n \Rightarrow F^nT^n$
denote the natural transformation obtained by repeatedly using the distributive law
\(\lambda\) to move all occurrences of \(T\) to the right.
\begin{lem} \label{lem:bel_pol}
  We have the following equalities: 
  \begin{enumerate}
    \item \label{item:bel_delta}
  $\big(\domofslice{T}{O}{\obsv} \times A \xrightarrow{(c^\Bel)^\dagger} FT(\domofslice{T}{O}{\obsv}) \xrightarrow{F(\flatten_\obsv)} FTS\big) \\ = \big(\domofslice{T}{O}{\obsv} \times A \xrightarrow{\iota_\obsv \times A} TS \times A \xrightarrow{\str^T} T(S \times A) \xrightarrow{T\delta^\dagger} TFTS \xrightarrow{\lambda_{TS}} FT^2 S \xrightarrow{F\mu_S^T} FTS \big)$.
    \item \label{item:bel_delta_eta} 
    $\big(S \times A \xrightarrow{\eta^{\slicemonad{T}{O}}_\obsv \times A} \domofslice{T}{O}{\obsv} \times A \xrightarrow{(c^\Bel)^\dagger} FT(\domofslice{T}{O}{\obsv}) \xrightarrow{F(\flatten_\obsv)} FTS\big) = \big(S \times A \xrightarrow{\delta^\dagger} FTS\big)$.
  \end{enumerate}
\end{lem}
See 
\ifarxiv
\Cref{ap:proof_lem_correct_flat} 
\else
\cite{DBLP:journals/corr/abs-2604-25355}
\fi
for the proof.

The following lemma, especially its second statement, is a key result to show the correctness theorem.
A standard way to prove such a preservation result would be to construct a coalgebra morphism, between $\Pol(c)$ and $\Pol(\langle c^\Bel, \slicemonad{T}{O}\obsv\rangle)$.
However,
we do not expect such a morphism to exist in general.
Instead, under the following assumptions $(\star)$:
\begin{enumerate}
  \item 
  the monotone algebra $\tau\colon FT\Omega \to \Omega$ is of the form $\rho \circ F\sigma$ for some monotone algebra $\rho \colon F\Omega \to \Omega$ and monotone Eilenberg--Moore algebra $\sigma \colon T\Omega \to \Omega$ such that 
  $\sigma \circ T(\rho) = \rho \circ F(\sigma) \circ \lambda_\Omega$ and $\bot_{TX} = \sigma \circ T(\bot_X)$ for each $X \in \cat{C}$;
  \item 
  $\big(\domofslice{T}{O}{\obsv} \xrightarrow{\iota_\obsv} TS \xrightarrow{T\langle \id, \obsv \rangle} T(S \times O)\big) = \big(\domofslice{T}{O}{\obsv} \xrightarrow{\langle \iota_\obsv, \slicemonad{T}{O}\obsv\rangle} TS \times O \xrightarrow{\str^T} T(S \times O)\big)$,
\end{enumerate}
we show that 
the morphism 
 $\mu\circ T(\str^T\circ(\iota_{\obsv}\times \id))\colon  T(\domofslice{T}{O}{\obsv} \times A^{O^*}) \to T(S \times A^{O^*})$,
becomes compatible with the transition structures after an application of 
\(FT\) and postcomposition with \(F^2\mu\circ F\lambda\).
See 
\ifarxiv
\Cref{ap:proof_correct}
\else
\cite{DBLP:journals/corr/abs-2604-25355}
\fi
 for the proof.

\begin{lem} \label{lem:commute_delta_mu}
  \begin{enumerate}
    \item \label{item:delta_eta} $\Pol(c) = F(\mu \circ T(\str^T \circ \iota_\obsv \times \id)) \circ \Pol(\langle c^\Bel, \slicemonad{T}{O}\obsv \rangle) \circ (\eta^{\slicemonad{T}{O}}_\obsv \times \id)$.
    \item \label{item:delta_mu}
  Under the assumption ($\star$),
  the following equation between morphisms of type $FT(\domofslice{T}{O}{\obsv} \times A^{O^*}) \to F^2T(S \times A^{O^*})$ holds:
  \begin{align*}
&F^2\mu \circ F\lambda \circ FT(\Pol(c)) \circ F(\mu \circ T(\str^T \circ \iota_\obsv \times \id)) \\
&= F^2 \mu \circ F\lambda \circ FTF(\mu \circ T(\str^T \circ \iota_\obsv \times \id)) \circ FT(\Pol(\langle c^\Bel, \slicemonad{T}{O}\obsv \rangle)).
  \end{align*}
  \end{enumerate}
\end{lem}

\begin{thm}\label{thm:belief-correctness}
Under the assumption $(\star)$,
the equation $\Vc{i, c}=\Vc{\Bel(i, c)}$ holds.
\end{thm}
\begin{proof}
  We prove that 
$\Phi_{\Pol(c)}^{n}(\bot)
=
\Phi_{\Pol(\langle c^\Bel, \slicemonad{T}{O}\obsv \rangle)}^{n}(\bot) \circ (\eta^{\slicemonad{T}{O}}_\obsv \times \id)$ for each $n \geq 1$.
  Write $\ol{T}S$ for the object $\domofslice{T}{O}{\obsv}$.
  By definition,
  $\Phi_{\Pol(c)}^{n}(\bot)
  =
  \tau^n\circ (FT)^n(\bot)\circ \Pol(c)^n = \tau^n\circ (FT)^n(\bot)\circ FT(\Pol(c)^{n-1}) \circ \Pol(c)$.
  Hence by \Cref{lem:commute_delta_mu}.\ref{item:delta_eta},
  it suffices to show 
  the following equality between 
  morphisms of type $FT(\ol{T}S \times A^{O^*}) \to \Omega$:
  \begin{align} 
  &\tau^{n} \circ (FT)^{n} \bot \circ FT(\Pol(c)^{n-1}) \circ F(\mu \circ T(\str^T \circ \iota_\obsv \times \id)) \notag \\
  &= \tau^{n} \circ (FT)^{n} \bot \circ FT(\Pol(\langle c^\Bel, \slicemonad{T}{O}\obsv \rangle)^{n-1}),\label{eq:ast_n}
  \end{align}
  for each $n \geq 1$.
  The full proof of this equality is deferred to 
  \ifarxiv
  \Cref{ap:diag_correctness_n}. 
\else
\cite{DBLP:journals/corr/abs-2604-25355}.
\fi
  Here we only give a proof sketch.
  We proceed by induction on $n$.
For the base case \(n=1\), the claim follows from the factorization assumption
\(\tau=\rho\circ F\sigma\), together with 
the equality
$\sigma\circ T\bot_X \circ \mu_X = \sigma\circ T\bot_{TX}$ for each $X \in \cat{C}$.
  The inductive step 
  uses
$\sigma \circ T\bot = \bot$, $\sigma \circ \mu = \sigma \circ T\sigma$,
and
  \Cref{lem:commute_delta_mu}.\ref{item:delta_mu}.
\end{proof}

\begin{exa}[nondeterminism]
  \label{ex:nondetAssumption}
  We continue~\cref{ex:nondetOb}. By~\cref{thm:belief-correctness}, it suffices to check the assumption ($\star$). 
  Set $\rho\colon \bset + \{\checkmark\} \rightarrow \bset$ and $\sigma\colon \pefunc(\bset)\rightarrow \bset$ by $\rho(b) = b$ for $b\in \bset$ and $\rho(\checkmark) = \btrue$, and $\sigma(U) = \land U$, respectively. 
  Clearly, we have $\tau = \rho \circ (\sigma + \{\checkmark\})$, and $\bot_{\pefunc X} = \sigma \circ \pefunc(\bot_X)$; note that the latter condition is not satisfied if we replace $\pefunc$ with $\pfunc$, and this is the reason why we employ this definition.  
  \ifarxiv
  See~\cref{ap:nondetAssump} 
\else
See~\cite{DBLP:journals/corr/abs-2604-25355}
\fi
  for the details of the remaining conditions. 
\end{exa}

\section{Comparing Semantics of Partially and Fully Observable Coalgebras}
\label{sec:reachable}
In this section, we compare the semantics of a pointed PO coalgebra and those of its fully observable counterpart, obtained by replacing the observation map with the identity.
We show that the semantics of a pointed PO coalgebra is always bounded above by that of its fully observable counterpart, 
and we give two sufficient conditions under which the two semantics coincide:
the first, given in \Cref{prop:partial-information-upper-bound}, is simpler, whereas the second, given in \Cref{prop:partial_upper_hist_new}, is weaker.

\begin{prop}\label{prop:partial-information-upper-bound}
For each pointed PO coalgebra 
$(i, \langle \delta,\obsv\rangle) \colon I \to S \to (FTS)^A \times O$,
it holds that
$\Vc{i, \langle \delta,\obsv\rangle}
\le
\Vc{i, \langle \delta,\id_S\rangle}$.
If $\obsv$ is split mono, then $\Vc{i, \langle \delta,\obsv\rangle}
=
\Vc{i, \langle \delta,\id_S\rangle}$ holds.
\end{prop}
This follows from the fact that \((\id_S,\obsv)\colon \langle \delta,\obsv\rangle \to
\langle \delta,\id_S\rangle\) is a morphism in the category \(\CoalgOctr{FT(\_)^A}\), together
with functoriality of \(\Pol\) and the preservation of the semantics under
\(FT\)-coalgebra morphisms shown in~\Cref{prop:Phi-coalg-morphism}. The split mono case is also proved similarly. The full
proof is given in
\ifarxiv
\Cref{ap:proof_fully_observable}.
\else
\cite{DBLP:journals/corr/abs-2604-25355}.
\fi

In \Cref{sec:semBeliefCoal}, we defined the semantics
\[
\Hc{i, c}
=
\bigvee_{u:O^+\to A,\;n\in\nat}
\Phi^n_{\Hist(c)_u}(\bot)\circ \nil_S\circ i
\]
as the join of the semantics of the pointed coalgebras $(\nil_S\circ i,\Hist(c)_u)$ for all $u\colon O^+ \to A$.
By \Cref{prop:Phi-coalg-morphism}, the semantics of a pointed $FT$-coalgebra does not change under pointed coalgebra morphisms. 
Hence, for each $u$, one may compute the
semantics of $(\nil_S \circ i,\Hist(c)_u)$ on any pointed subcoalgebra.
\begin{lem} \label{lem:subcoalg_u}
  For each morphism $(f, g)$ from $c\colon S \to (FTS)^A \times O$ to $c'\colon S' \to (FTS')^A \times O'$ in $\CoalgOcov{FT(\_)^A}{}$ and each morphism $u\colon O' \to A$, 
  the morphism $f$ is also a morphism $f\colon c_{u \circ g} \to c'_u$ in $\Coalg{FT}$.
\end{lem}
The following proposition is based on the observation that, by
\Cref{lem:subcoalg_u} and \Cref{prop:Phi-coalg-morphism}, 
for each scheduler,
the history-based semantics 
may be computed on 
a suitable pointed subcoalgebra
of the corresponding resolved history coalgebra. 
See 
\ifarxiv
\Cref{ap:proof_partial_upper_hist_new} 
\else
\cite{DBLP:journals/corr/abs-2604-25355}
\fi
for the proof.

\begin{prop} \label{prop:partial_upper_hist_new}
  Let $(i, \langle \delta,\obsv\rangle) \colon I \to S \to (FTS)^A \times O$ be an $I$-pointed PO coalgebra.
  Suppose that, for each $u\colon S^+ \to A$, there exists an $I$-pointed $FT$-subcoalgebra 
  $m_u\colon (i_u, \delta_u) \rightarrowtail (\nil \circ i, \Hist(\langle \delta, \id\rangle)_u)$
  such that there is $f_u\colon O^+ \to S^+$ satisfying $f_u \circ \obsv^+ \circ m_u = m_u$.
  Then the equation $\Hc{i, \langle \delta, \obsv\rangle} = \Hc{i, \langle \delta, \id\rangle}$ holds.
\end{prop}
Note that the split mono condition in \Cref{prop:partial-information-upper-bound} implies the condition in \Cref{prop:partial_upper_hist_new}:
if $\obsv$ is split mono, then $\obsv^+$ is also, and hence $m_u \coloneqq \id_{(\nil \circ i, \Hist(\langle \delta, \id \rangle)_u)}$ satisfies the condition.

\begin{rem}[Reachable coalgebras] \label{rem:reachable-part}
  When $\cat{C} = \Set$, 
  the equality $\Hc{i, \langle \delta, \obsv\rangle} = \Hc{i, \langle \delta, \id\rangle}$ in \Cref{prop:partial_upper_hist_new} may be established pointwise in $x \in I$.
  Indeed, regarding $x \in I$ as a morphism $x \colon 1\to I$, we have $\Hc{i, c}(x) = \Hc{i \circ x, c}(*)$,
  where $* \in 1$ is the unique element.
  Since joins in $\Set(I, \Omega)$ are computed pointwise, it therefore suffices, for each $x \in I$, to apply \Cref{prop:partial_upper_hist_new} to the $1$-pointed coalgebra whose initial point is $i \circ x\colon 1\to S$.

    For fixed $x \in I$ and $u\colon S^+ \to A$, 
    a canonical choice for the required subcoalgebra is 
  the
  \emph{reachable coalgebra}~\cite{DBLP:journals/corr/abs-1305-0576},
  a pointed coalgebra having no proper pointed subcoalgebra, whenever it exists.
A reachable subcoalgebra of a given pointed coalgebra is also called its \emph{reachable part}. 
This notion coincides with the usual reachable part of the graph induced by the given coalgebra when $\cat{C} = \Set$ and the functor under consideration preserves intersections~\cite{Wissmann2019}.
There are several constructions of a reachable subcoalgebra of a given pointed coalgebra in the literature.
For example, 
it is known that
if
$\cat{C}$ has intersections and
 a functor $G\colon \cat{C} \to \cat{C}$ preserves them,
 then each $G$-coalgebra has a unique reachable subcoalgebra,
 obtained as the intersection of all subcoalgebras~\cite{DBLP:journals/corr/abs-1305-0576}.
 Another construction of a reachable subcoalgebra is proposed in~\cite{Wissmann2019} via an iterative computation.

 In all concrete applications of \Cref{prop:partial_upper_hist_new} in this paper, we work pointwise in $x \in I$ and take the subcoalgebra to be the reachable part.
\end{rem}

These results are useful when comparing the semantics $\Vc{i, c}$ of a pointed PO coalgebra and that of the fully observable counterpart of the belief coalgebra.
\Cref{thm:belief-correctness}
yields the equality 
$\Vc{i, c} = \Vc{\Bel(i, c)}$,
where
the right-hand side is still formulated as the semantics of a pointed PO coalgebra.
By applying \Cref{prop:partial_upper_hist_new} to the belief coalgebra $\Bel(i, c)$, 
we obtain a sufficient condition 
to ensure that
the semantics of $\Bel(i, c)$ can be replaced by that of the coalgebra 
in the fully
observable setting:
$\Vc{i, c} = \Vc{\eta^{\slicemonad{T}{O}}_\obsv \circ i, \big\langle c^\Bel, \id\big\rangle}$.

\begin{exa}[nondeterminism]
In the case of nondeterministic systems, the reachable part (cf.~\Cref{rem:reachable-part}) satisfies the condition in \Cref{prop:partial_upper_hist_new}, since from an initial state, the mapping from paths over its history coalgebra under a scheduler $u\colon S^+ \to A$ to their observation sequences is injective by the definition of the belief construction.
Since the belief coalgebra induced by a finite-state partially observable nondeterministic transition system is a finite-state labelled transition system, and since deciding the existence of a scheduler that ensures termination is clearly decidable for fully observable finite-state labelled transition systems, it follows from~\cref{thm:belief-correctness} that the termination problem for partially observable nondeterministic systems is decidable as well; see
\ifarxiv
\cref{prop:decidable} in~\cref{sec:omitReachable}
\else
\cite[Prop.~46]{DBLP:journals/corr/abs-2604-25355}
\fi
 for the details. 
\end{exa}

\section{Examples} \label{sec:examples}
We instantiate our framework with (1) partially observable Markov decision processes (POMDPs) and (2) partially observable weighted transition systems.

\subsection{POMDP}
  \label{ex:POMDP}
Let $T$ be the finitely supported distribution monad  $\dfunc$ on $\Set$, and let $F \coloneqq  \_ \times \R$, where the $\R$-component represents the reward assigned to state-action pairs.
We define a \emph{distributive law} $\lambda_X\colon \dfunc(X \times \R) \rightarrow \dfunc(X) \times \R$ by
\begin{align*}
  \lambda_X(\nu) \coloneqq \Big(\dfunc(\pi_1)(\nu), \sum_{(x, r)\in \supp(\nu)} \nu(x, r) \cdot r\Big).
\end{align*}

For a morphism $f\colon X\rightarrow O$, 
the object $\dom \slicemonad{\dfunc}{O}(f)$ consists of beliefs $b \in \dfunc(X)$ whose support is contained in a single fibre of $f$;
concretely, there exists a (unique) $o \in O$ such that $\supp(b) \subseteq f^{-1}(o)$.
Define a belief decomposition $\{\alpha^O_f\colon (\dfunc \dom)(f) \rightarrow  (\dfunc\dom\,\ol{\dfunc})(f)\}_{f\colon X \to O}$ by
\begin{align*}
  \alpha^O_f(b)(b') \coloneqq \begin{cases}
    Z_o(b) & \text{if $\exists o\in O.~Z_o(b) > 0$ and $b' = \frac{1}{Z_o(b)} \cdot \restr{b}{o}$, }\\
    0 &\text{otherwise, }
  \end{cases}
\end{align*}
 where $Z_o(b) \coloneqq \sum_{x\in f^{-1}(o)}b(x)$ and $\restr{b}{o}$ is defined by $\restr{b}{o}(x) \coloneqq b(x)$ if $f(x) = o$, and  $\restr{b}{o}(x) \coloneqq 0$ otherwise. 
The morphism \(\alpha_f^O\) decomposes the belief \(b\) into the family of conditional
beliefs indexed by observations; see~\cref{app:omitDPpomdp} for the details of the belief decomposition $\alpha$.

  We define the ordered object $\Omega$ to be $\eR$ with the standard total order, and we equip each homset $\Set(X, [0, \infty])$ with the pointwise order.
  We define the two algebras $\rho \colon \eR\times \R\rightarrow \eR $ and $\sigma\colon  \dfunc(\eR) \rightarrow \eR $ by 
  \begin{align*}
    \rho(r_1, r_2)\coloneqq r_1 + r_2,\text{ and }\sigma(\nu) \coloneqq \sum_{r\in \supp(\nu)}\nu(r) \cdot r.   
  \end{align*}
  The resulting algebra $\tau \colon \dfunc(\eR)\times \R\rightarrow \eR$ obtained by $\rho$ and $\sigma$ is thus given by 
$\tau(\nu, r) = r + \sum_{r'\in \supp(\nu)}\nu(r') \cdot r'$.

  Given $c\colon S \to  \dfunc(S)\times \R$, the monotone map $\Phi_{c} \colon \cat{C}(S,\eR) \to \cat{C}(S,\eR)$
  is 
  \begin{align*}
    \Phi_{c}(f)(s) = (\pi_2\circ c)(s) + \sum_{s'\in S} (\pi_1\circ c)(s)(s')\cdot f(s').
  \end{align*}

  Now let
  $c = \langle \delta,\obsv\rangle\colon S \to \big(\dfunc(S)\times \R\big)^A \times O$.
  Then for each  $u\colon O^{+}\rightarrow A$ and $n \in \nat$, 
  the composite $\Phi_{\Pol(c)}^n(\bot)\langle \id, (u \circ \cons)^\dagger \circ \obsv \rangle\colon S\rightarrow \eR$ maps $s \in S$ to
  \begin{align*}
    \Phi_{\Pol(c)}^n(\bot)(s, \tilde u_s) =  \text{the total expected reward from $s$ in $n$ steps under $u$,}
  \end{align*}
  where $\tilde u_s \coloneqq \lambda \vec{o}.u(\obsv(s) \vec{o}) \colon O^* \to A$.
  Consequently, given an initial element $x\in I$, the objective $\Vc{i, c}$ is given by 
  \begin{align*}
    \Vc{i, c}(x) = \text{the maximum total expected reward from $i(x)$ over schedulers $u\colon O^{+}\rightarrow A$.}
  \end{align*}
The belief construction yields the usual belief MDP $c^\Bel$ together with an
observation map $\slicemonad{\dfunc}{O}(\obsv)$ sending
each
belief to the unique observation on whose fibre its support lies.

For the belief coalgebra 
$\Bel(i, c)\colon I  \to \domofslice{\dfunc}{O}{\obsv} \to \big(\dfunc(\domofslice{\dfunc}{O}{\obsv})\times \R\big)^A\times O$,
the reachable part chosen as in~\Cref{rem:reachable-part} for each initial element and scheduler
satisfies the condition in \Cref{prop:partial_upper_hist_new}.
Because for each belief, action, and observation, there is at most
one successor belief carrying that observation,
every reachable belief history 
under a
fixed scheduler $u$
is uniquely
determined by its observation history.

Consequently, for POMDPs, the semantics of the belief coalgebra agrees with the
semantics of the corresponding fully observable belief MDP.
Combining this with \Cref{thm:belief-correctness}, we recover the standard equivalence between the
semantics of a POMDP and that of its fully observable belief MDP.

\subsection{Partially Observable Weighted Transition Systems}
\label{subsec:exweight}
Finally, we illustrate our belief construction for partially observable weighted transition systems.  

  Let $T\coloneqq \rfunc$ be the semimodule monad~\cite{BonchiS21} that is induced by the standard semiring $(\R, + , \cdot)$ with the summation $+$ and the multiplication $\cdot$ on $\Set$, and $F = \_\times \R \colon \Set \to \Set$. 
  We have a distributive law $\lambda\colon \rfunc(\_\times \R)\Rightarrow \rfunc(\_)\times\R $ defined by
  \begin{align*}
    \lambda(\nu) \coloneqq \Big(\sum_{r\in \R} \nu(\_, r), \sum_{(x, r)\in X\times \R} \nu(x, r)\cdot r\Big).
  \end{align*} 
  For a morphism $f\colon X \to O$,
  the object $\dom \ol{\rfunc}(f)$ consists of 
  normalized finitely supported weight functions whose support is contained in a single fibre of $f$;
  that is, $\nu \in \mathcal{R}(X)$  such that $\sum_{x \in X}\nu(x) = 1$ and $\supp(\nu) \subseteq f^{-1}(o)$ for some (necessarily unique) $o \in O$.
  We define a belief decomposition $\{\alpha^O_f\colon (\rfunc \dom)(f) \to (\rfunc\dom\ol{\rfunc})(f)\}_{f\colon X \to O}$ by
 \[
  \alpha^O_f(\nu)(\mu) \coloneqq \begin{cases*}
    Z_o(\nu) &\text{ if  $\exists o \in O.~Z_o(\nu) > 0 \text{ and }\mu = \frac{1}{Z_o(\nu)}\cdot \restr{\nu}{o}$}, \\
    0 &\text{ otherwise, }
  \end{cases*}
 \]
 where $Z_o(\nu) \coloneqq \sum_{x\in f^{-1}(o)} \nu(x)$ and $\restr{\nu}{o}(x) \coloneqq \nu(x)$ if $f(x) = o$, and  $\restr{\nu}{o}(x) \coloneqq 0$ otherwise. One can show that $\alpha$ is a belief decomposition by essentially the same argument as in~\cref{ex:POMDP}.

We take the ordered object $\Omega$ to be $[0, \infty]$ with the standard total order, and we equip each $\Set(X, [0, \infty])$ with the pointwise order.
We define the two algebras $\rho \colon [0, \infty]\times \R\rightarrow[0, \infty] $ and $\sigma\colon  \rfunc([0, \infty]) \rightarrow [0, \infty] $ by 
$\rho(r_1, r_2) \coloneqq r_1 + r_2$, and $\sigma(\nu) \coloneqq \sum_{r\in [0, \infty]} \nu(r)\cdot r$, for any $\nu\in \rfunc([0, \infty])$.  
The resulting algebra $\tau\colon \rfunc([0, \infty])\times \R\rightarrow [0,\infty]$ is then given by $\tau(\nu, r) = r + \sum_{r'\in [0, \infty]} \nu(r') \cdot r'$. 
Given a coalgebra $c\colon S\rightarrow \rfunc(S)\times \R$, we regard $c$ as a weighted transition system 
where $\pi_1(c)\colon S\rightarrow \rfunc(S)$ describes the weight assigned to transitions, and $\pi_2(c)\colon S\rightarrow \R$ describes the weight to terminate immediately; 
see~\cref{sec:omitDefExample} for the details. We remark that we can show that the assumption $(\star)$ holds by almost the same argument as in~\cref{ex:POMDP}.

The monotone map $\Phi_{c} \colon \cat{C}(S,[0, \infty]) \to \cat{C}(S,[0, \infty])$ is then given by
  \begin{align*}
  \Phi_{c}(f)(s) =  (\pi_2\circ c)(s)+ \sum_{s'\in S} (\pi_1\circ c)(s)(s')\cdot f(s').
\end{align*}

Given a coalgebra $c = \langle \delta,\obsv\rangle\colon S \to \big(\rfunc(S)\times \R\big)^A \times O$
and a scheduler $u\colon O^{+}\rightarrow A$, 
the composite $\Phi_{\Pol(c)}^n(\bot)\langle \id, (u \circ \cons)^\dagger \circ \obsv \rangle\colon S\rightarrow [0, \infty]$ maps $s \in S$ to
\begin{align*}
  \Phi_{\Pol(c)}^n(\bot)(s, \tilde u_s) 
  = \parbox[t]{0.58\linewidth}{%
    the sum of weights over all terminating paths from $s$
    in $n$ steps under $u$.%
  }
\end{align*}
  where $\tilde u_s \coloneqq \lambda \vec{o}.u(\obsv(s) \vec{o}) \colon O^* \to A$.
See~\cref{sec:omitDefExample} for the definitions of terminating paths and their weights. 
Consequently, given an initial element $x\in I$ in $I$-pointed coalgebra $(i, c)\colon I\rightarrow S \to \big(\rfunc(S)\times \R\big)^A \times O$, the objective $\Vc{i, c}$ is given by 
\begin{align*}
  \Vc{i, c}(x) = \text{the maximum sum of weights over terminating paths from $i(x)$ over $u$.}
\end{align*}
By \Cref{thm:belief-correctness} and
the same argument as in \Cref{ex:POMDP},
the objective $\Vc{i, c}$ 
coincides with 
that of the fully observable counterpart of $\Bel(i, c)$.

\section{Related Work}
\label{sec:relatedwork}
The work most closely related to ours is the recent work by Baltieri, Torresan, and Nakai~\cite{baltieri2025coalgebraic}, which proposes a coalgebraic approach to POMDPs in order to capture notions of behavioural equivalence in a systematic manner. They also propose a generalized determinization of POMDPs `a la~\cite{DBLP:journals/corr/abs-1302-1046}. This differs from our belief construction: their generalized determinization yields deterministic systems, whereas our belief coalgebras remain effectful, reflecting the stochastic nature of belief MDPs.
Compared with their work, our novelty lies in providing a unified coalgebraic belief construction and proving the coincidence of two semantics: one for partially observable systems and the other for belief coalgebras. Such an equivalence between partially observable systems and their belief coalgebras has not been established in~\cite{baltieri2025coalgebraic}.

Bonchi, Sokolova, and Vignudelli~\cite{BonchiSV22} propose a determinization  for systems combining nondeterministic and probabilistic choices.
In their work, the monad of convex subsets of distributions plays a crucial role. Subsequently, Goy and Petrisan showed that this monad arises from a weak distributive law of the powerset monad over the finite distribution monad.
Goy~\cite{Goy21} showed that the belief-state transformer of probabilistic automata~\cite{BonchiSS21} can be derived naturally from weak distributive laws, and Turkenburg et al.~\cite{TurkenburgKRS23} developed the notion of invertible steps~\cite{RotJL21} induced by weak distributive laws to show the preservation and reflection of bisimilarity.
Future work is to investigate whether our belief construction can be extended to these generalized determinization constructions arising from weak distributive laws.

Bezhanishvili, Cupke, and Panangaden~\cite{BezhanishviliKP12} study minimization in the category of compact Hausdorff spaces to support the minimization of belief automata for POMDPs.
In this paper, we study the construction of belief coalgebras from pointed PO coalgebras and the equivalence of two semantics that were not considered in~\cite{BezhanishviliKP12}.

Jacobs and Sokolova~\cite{JacobsS09} study the notion of schedulers coalgebraically, and they propose a coalgebraic treatment of schedulers for nondeterministic systems through a strong monad map. 
Our treatment of strategies is different and we regard executions under strategies as stateful computations (see~\cref{def:pol}).

\section{Conclusion}
\label{sec:conclusion}

We propose a coalgebraic belief construction for partially observable systems, including POMDPs, and show the equivalence of two semantics: that of a partially observable system and that of its belief coalgebra, under a mild assumption.
As future work, we would like to support randomized strategies, which are commonly used in POMDPs.
More generally, supporting effectful strategies itself would be an interesting direction for future work.
\newpage
\bibliography{mybib}

\appendix
\newpage

\section{An Indexed-Categorical View of \Cref{def:alpha}} \label{ap:indexed-belief-decomp}

\begin{rem}
  By precomposing the functor $\Sl$ from \Cref{rem:indexed_ol_T} with
  the inclusion functor $\cat{O} \hookrightarrow \cat{C}$, we may regard it 
  as a functor $\cat{O} \to \Cat$.
Let $\Delta \cat{C}\colon \cat{O} \to \Cat$ be the constant functor at $\cat{C}$,
and let $\dom\colon \mathrm{Sl} \Rightarrow \Delta \cat{C}$ be the indexed
functor whose component at $O$ is $\dom\colon \cat{C}/O \to \cat{C}$.
The monads $T$ on $\cat{C}$ and $\slicemonad{T}{O}$ on slice categories
induce
the indexed functor $T\colon \Delta \cat{C} \Rightarrow \Delta \cat{C}$ and
the lax indexed functor
$\ol{T}\colon \Sl \Rightarrow \Sl$, respectively.
With this notation,
the second condition in \Cref{def:alpha} says exactly that the family
$\alpha=(\alpha^O)_{O \in \cat{O}}$ forms a modification
$\alpha\colon T \dom \Rightarrow T \dom\,\ol{T}$.
Likewise, 
the family $\flatten=(\flatten^O)_{O \in \cat{O}}$ forms a modification
$\flatten \colon T\dom \ol{T} \Rightarrow T\dom$.
Hence a belief decomposition structure on $T$ can be seen simply as a
modification $\alpha \colon T\dom \Rightarrow T\dom \ol{T}$ 
that is a section of $\flatten$.
\end{rem}

\section{Omitted Definitions and Proofs in~\cref{sec:examples}}
\label{sec:omitDefExample}

\subsection{Omitted Definitions and Proofs in~\cref{ex:POMDP}}
\label{app:omitDPpomdp}
We first see that $(\alpha^O\colon \dfunc \dom \Rightarrow  \dfunc\dom\,\ol{\dfunc})_O$ is a belief decomposition. 
\begin{proposition}
  The data $(\alpha^O\colon \dfunc \dom \Rightarrow  \dfunc\dom\,\ol{\dfunc})_O$ is a belief decomposition. 
\end{proposition}
\begin{proof}
  We first see that $\alpha^O$ is a natural transformation. 
  Let $h\colon f\rightarrow g$ and $b\in \dfunc(X)$. We have 
  \begin{align*}
   &\big((\dfunc\dom\,\ol{\dfunc})(h)\circ  \alpha^O_f\big)(b)(b')\\
    = &\begin{cases}
    \sum_{x\in f^{-1}(o)}b(x) & \text{if $\exists o\in O.~\sum_{x\in f^{-1}(o)}b(x) > 0$ and $b' = \frac{1}{\sum_{x\in f^{-1}(o)}b(x)} \cdot B(o)$, }\\
    0 &\text{otherwise, }
  \end{cases}\\
  &\big(\alpha^O_g\circ (\dfunc \dom)(h)\big)(b)(b')\\
  = &\begin{cases}
    \sum_{y\in g^{-1}(o)}\dfunc(h)(b)(y) & 
    \text{if $\exists o\in O.~C(o) > 0$ and $b'(y) = \frac{1}{C(o)} \cdot \restr{\dfunc(h)(b)}{o}$, }\\
    0 &\text{otherwise, }
  \end{cases}\\
  = &\begin{cases}
    \sum_{x\in f^{-1}(o)}b(x) & \text{if $\exists o\in O.~\sum_{x\in f^{-1}(o)}b(x) > 0$ and $b' = \frac{1}{\sum_{x\in f^{-1}(o)}b(x)} \cdot \restr{\dfunc(h)(b)}{o}$, }\\
    0 &\text{otherwise, }
  \end{cases}\\
  = &\begin{cases}
    \sum_{x\in f^{-1}(o)}b(x) & \text{if $\exists o\in O.~\sum_{x\in f^{-1}(o)}b(x) > 0$ and $b' = \frac{1}{\sum_{x\in f^{-1}(o)}b(x)} \cdot B(o)$, }\\
    0 &\text{otherwise, }
  \end{cases}
  \end{align*}
  where 
  $B(o)(y) = \sum_{x\in f^{-1}(o)\cap h^{-1}(y)} b(x)$ and
  $C(o) = \sum_{y\in g^{-1}(o)}\dfunc(h)(b)(y)$.
  It satisfies the condition (i) in~\cref{def:alpha} as follows: 
  \begin{align*}
  (\flatten^O_f \circ \alpha^O_f)(b)(x) &= \sum_{b'} \alpha^O_f(b)(b')\cdot b'(x) \\
  &= \big(\sum_{x'\in X\text{ s.t. }f(x') = f(x)}b(x')\big)\cdot\frac{b(x)}{\sum_{x'\in X\text{ s.t. }f(x') = f(x)}b(x')} = b(x). 
  \end{align*}
  It is easy to see that the condition (ii) in~\cref{def:alpha} holds for $(\alpha^O)_O$. 
\end{proof}
Next, we see that the two modalities $\rho \colon \eR\times \R\rightarrow \eR $ and $\sigma\colon  \dfunc(\eR) \rightarrow \eR $  satisfy the first assumption in $(\star)$. 
\begin{proposition}
  The morphisms $\rho$ and $\sigma$ satisfy the first assumption in $(\star)$.
\end{proposition}
\begin{proof}
It is straightforward to see that $\rho$ is a monotone algebra, and $\sigma$ is a monotone Eilenberg-Moore algebra. 
We see that $\bot_{\dfunc(X)} = \sigma\circ \dfunc(\bot_X)$ in the following: 
\[
\big(\sigma\circ\dfunc(\bot_X)\big)(\nu) = \sum_{x\in X} \nu(x) \cdot \bot_X(x) = 0 = \bot_{\dfunc(X)}(\nu).
\]
We next see that $\sigma \circ \dfunc(\rho) = \rho \circ (\sigma\times \R) \circ \lambda_\Omega$ as follows: 
\begin{align*}
  \big(\sigma \circ \dfunc(\rho)\big)(\nu) &= \sum_{r}  \dfunc(\rho)(\nu)(r) \cdot r = \sum_{r_1, r_2}  \nu(r_1, r_2) \cdot (r_1 + r_2) \\
  \big(\rho \circ (\sigma\times \R) \circ \lambda_\Omega \big)(\nu) &= \big(\pi_2\circ \lambda_{\Omega}\big)(\nu) + \sum_{r}(\pi_1 \circ \lambda_{\Omega})(\nu)(r)\cdot r \\
  &=\big(\sum_{r_1, r_2} \nu(r_1, r_2)\cdot r_2 \big) + \big(\sum_{r} \big(\sum_{r'}\nu(r, r')\big)\cdot r\big)\\
  &= \sum_{r_1, r_2} \nu(r_1, r_2)\cdot (r_1 + r_2). 
\end{align*}
\end{proof}

Lastly, we check the second assumption in $(\star)$. 
\begin{proposition}
  The second assumption in $(\star)$ holds for POMDPs.
\end{proposition}
\begin{proof}
  We have
\begin{align*}
  \big(\dfunc\langle \id, \obsv \rangle\circ \iota_\obsv\big)(\nu)(x, o) &= \begin{cases}
    \nu(x) &\text{ if $o = \obsv(x)$,}\\
    0 &\text{otherwise,}
  \end{cases} \\
  \big(\str^{\dfunc} \circ \langle \iota_\obsv, \slicemonad{\dfunc}{O}\obsv\rangle\big)(\nu)(x, o) &=\big(\str^{\dfunc}\big(\nu, o' \big)\big)(x, o) = \begin{cases}
    \nu(x) &\text{ if $o = \obsv(x)$,}\\
    0 &\text{otherwise,}
  \end{cases}
\end{align*}
where $o' = \obsv(x')$ for some $x'\in \supp(\nu)$. 
\end{proof}
\subsection{Omitted Definitions in~\cref{subsec:exweight}}
Given a coalgebra $c\colon S\rightarrow \rfunc(S)\times \R$, we define a weighted transition system $c'\colon S\rightarrow \rfunc(S + \{\checkmark\})$ such that $c'(s)(s') \coloneqq (\pi_1\circ c)(s)(s')$ and $c'(s)(\checkmark) \coloneqq (\pi_2\circ c)(s)$. 
A \emph{terminating path} $p$ on $c$ is the path $p$ on $c'$ such that $p$ ends at $\checkmark$. 
The \emph{weight} $w(p)$ over the terminating path $p = s_1\cdots s_n \cdot \checkmark$ is given by  
\begin{align*}
  w(p) \coloneqq \big(\prod_{i\in [1, n-1]} c'(s_i, s_{i+1})\big)\cdot c'(s_{n}, \checkmark). 
\end{align*}
The characterization of the objective $\Vc{i, c}$ given in~\cref{subsec:exweight} is a corollary of the following characterization.  

\begin{proposition}
  Let $c\colon S\rightarrow \rfunc(S)\times \R$, and $\Phi_{c} \colon \cat{C}(S,[0, \infty]) \to \cat{C}(S,[0, \infty])$ be the monotone map defined in~\cref{subsec:exweight}. 
  We have 
  \begin{align*}
 \bigvee_{n\in \N} \Phi_{c}^n(\bot)(s) =  \sum_{p\in \mathrm{TPath}(s)} w(p),
  \end{align*}
  where $\mathrm{TPath}(s)$ is the set of terminating paths starting from $s$.
  
\end{proposition}

\section{Omitted Proofs} 

\subsection{Proof of~\Cref{prop:ol_T}}\label{prop:appol_T}
\begin{proof}
The functor $\slicemonadO{T}{O}$ forms a monad, which comes from the sliced adjunction 
of the Eilenberg--Moore adjunction $F_T \dashv U_T\colon \cat{C}^T \to \cat{C}$ at $O$:
\begin{displaymath}
  \xymatrix@C=1em{
    \cat{C}/O \ar@/^1em/[rr]^{\slicefunc{(F_T)}{O}} &\bot &\cat{C}^T/F_T(O) \ar@/^1em/[ll]^{\slicefunc{(U_T)}{O}}
  }
\end{displaymath}
Here $\slicefunc{(F_T)}{O}$ is the functor induced by applying $F_T$ to the objects and morphisms of $\cat{C}/O$, and $\slicefunc{(U_T)}{O}$ is the functor defined by applying $U_T$ and taking the pullback along the unit of the adjunction $F_T \dashv U_T$.
Therefore, the induced monad $\slicefunc{(U_T)}{O}\slicefunc{(F_T)}{O}$ is $\slicemonadO{T}{O}$.

Let $u\colon O \to O'$.
Then we have the equality $\slicefunc{(F_T)}{O'} \circ \Sigma_u = \Sigma_{F_T(u)} \circ \slicefunc{(F_T)}{O}$.
Hence 
by applying \Cref{lem:adj_monad} to 
the functors $\Sigma_u$, $\Sigma_{F_T(u)}$ and
the adjunctions $\slicefunc{(F_T)}{O'} \dashv \slicefunc{(U_T)}{O'}$ and $\slicefunc{(F_T)}{O} \dashv \slicefunc{(U_T)}{O}$,
the statement about the monad morphism follows.
The explicit description of the natural transformation $\theta_u$ follows from the construction of the monad morphism in \Cref{lem:adj_monad}.
\end{proof}

\subsection{Proof of~\Cref{prop:iota}}\label{ap:proof_iota}

\begin{proof}
  Apply \Cref{lem:adj_monad} to the 
  forgetful functors $\cat{C}/O \to \cat{C}$ and $\cat{C}^T/F_T(O) \to \cat{C}^T$ and the adjunctions $F_T \dashv U_T$ and $\slicefunc{(F_T)}{O}\dashv \slicefunc{(U_T)}{O}$.
\end{proof}

\subsection{Proof of~\Cref{prop:bel}}\label{ap:proof_bel}
\begin{proof}
  Let $h \coloneqq \dom\slicemonad{T}{O'}(f) \circ \dom(\theta_g)_{\obsv}$.
  We show that 
\((h,g)\) is a morphism of \(I\)-pointed PO coalgebras. 
First, since \((\Sigma_g,\theta_g)\) is a monad morphism, its compatibility with units
yields
$h\circ \eta^{\slicemonad{T}{O}}_{\obsv}
=
\eta^{\slicemonad{T}{O'}}_{\obsv'}\circ f$.
Hence
$h\circ \eta^{\slicemonad{T}{O}}_{\obsv}\circ i
=
\eta^{\slicemonad{T}{O'}}_{\obsv'}\circ f\circ i
=
\eta^{\slicemonad{T}{O'}}_{\obsv'}\circ i'$.

Next, by construction, \(h\) is a morphism
$\Sigma_g(\slicemonad{T}{O}(\obsv)) \to \slicemonad{T}{O'}(\obsv')$
in the slice category \(\cat C/O'\). Therefore the observation part commutes:
$g\circ \slicemonad{T}{O}(\obsv)=\slicemonad{T}{O'}(\obsv')\circ h$.

The transition part also commutes by the following equations:
\[
\begin{aligned}
FTh \circ c^\Bel
&=FTh \circ F\alpha_{\obsv}\circ \Det(\delta)\circ \iota_{\obsv} \\
&=F\alpha_{\obsv'}\circ FTf\circ \Det(\delta)\circ \iota_{\obsv} \\
&=F\alpha_{\obsv'}\circ \Det(\delta')\circ Tf\circ \iota_{\obsv} \\
&=F\alpha_{\obsv'}\circ \Det(\delta')\circ \iota_{\obsv'}\circ h \\
&={\delta'}^{\Bel}\circ h .
\end{aligned}
\]
Here the second equality uses the compatibility of \(\alpha\) with base change together
with the naturality of \(\alpha^{O'}\), the third equality uses the functoriality of
ordinary coalgebraic determinization, and the fourth equality holds by definition of $h$.
\end{proof}

\subsection{Proof for $\alpha$ in~\cref{ex:nondet}}
\label{ap:proofAlphaNondet}
By definition, we can see that $(\ol{\pefunc})_O\colon \Set/O \to \cat \Set/O$ is defined by for any $f\colon S\rightarrow O$,
  $(\ol{\pefunc})_O(f)\colon \domofslicenew{T}{O}{f}\rightarrow O$ and $(\ol{\pefunc})_O(f)(U) = f(u)$ for some $u\in U$, where $\domofslicenew{T}{O}{f}\coloneqq \{U\in \pefunc(S)\mid \exists o\in O \text{ s.t. } U\subseteq f^{-1}(o)\}$. 
  For any $h\colon f\rightarrow g$ in $\Set/O$, we have $(\ol{\pefunc})_O(h)(U) = h(U)$.
  
We see that $\alpha^O$ is a natural transformation: given  $h\colon f\rightarrow g$ in $\Set/O$ and $U\in (\pefunc \dom)(f)$, we have 
\begin{align*}
  \big(\alpha^O_g \circ (\pefunc \dom)(h)\big)(U) &= \big\{  \{ v\in h(U)\mid g(v) = o\} \, \big| \,  o\in O\text{ s.t. }g^{-1}(o)\cap h(U)\not = \emptyset\big\}  \\
&= \big\{  \{ v\in h(U)\mid g(v) = o\} \, \big| \,  o\in O\text{ s.t. }f^{-1}(o)\cap U\not = \emptyset\big\}  \\
&= \big\{  \{ h(u)\mid u\in U\text{ and }f(u) = o\} \, \big| \,  o\in O\text{ s.t. }f^{-1}(o)\cap U\not = \emptyset\big\}  \\
&= \big((\pefunc\dom\ol{\pefunc})(h)\circ \alpha^O_f\big)(U).
\end{align*}
It is straightforward to see that $\alpha^O$ is a section of $\flatten^O$, since $\alpha^O$ simply creates the partition based on observations.

  Next,
  we show that $\alpha$ satisfies the condition (ii) in \Cref{def:alpha}.
  For each injective function $u\colon O \to O'$ and 
  each function
  $f\colon X \to O$,
  the component
  $(\pefunc \dom_{O'}\theta_u \circ \alpha^{O})_f$ is the function mapping $U \in \pefunc X$ to $\{f^{-1}(o) \cap U \mid o \in O \text{ s.t.~}f^{-1}(o) \cap U \neq \emptyset\}$,
  and the component
  $(\alpha^{O'} \Sigma_u)_f$ is the function mapping $U \in \pefunc X$ to $\{(u \circ f)^{-1}(o') \cap U \mid o' \in O' \text{ s.t.~}(u \circ f)^{-1}(o') \cap U \neq \emptyset\}$.
  They are equal because $u$ is injective.

\subsection{Proof of~\Cref{prop:bel_delta_id}} \label{ap:proof_bel_id}
\begin{proof}
  Note that $f \coloneqq \slicemonad{T}{S}(\id)\colon \domofslice{T}{S}{\id_S} \to S$
  is an isomorphism because it is obtained as the pullback of the isomorphism $\id_{TS}$.
From the pullback square we have
$\iota_{\id_S} = \eta_S \circ f$.

By the assumption
$\mu_S \circ T\iota_{\id_S} \circ \alpha_{\id_S} = id_{TS}$, 
the following equations hold:
\[
\id_{TS}
= 
\mu_S \circ T\iota_{\id_S} \circ \alpha_{\id_S}
=
\mu_S \circ T\eta_S \circ Tf \circ \alpha_{\id_S}
=
Tf \circ \alpha_{\id_S}.
\]
By the following equations, 
$f\colon c^\Bel \to \delta$ is an $FT$-coalgebra morphism:
\[
FTf \circ c^\Bel
=
FTf \circ F\alpha_{\id_S} \circ \Det(\delta) \circ \iota_{\id_S}
=
\Det(\delta) \circ \iota_{\id_S}
=
\Det(\delta) \circ \eta_S \circ f
= \delta \circ f.
\]
Thus
$(f, \id_S) \colon \Bel(\id_S, \langle \delta,id_S\rangle) \to (\id_S, \langle \delta,\id_S\rangle)$
is an isomorphism in $\CoalgOcov{FT}{S}$. 
\end{proof}

\subsection{Proof of \Cref{prop:history-semantics}} \label{ap:proof_history_semantics}

\begin{proof}
  By \Cref{prop:Phi-coalg-morphism},
  it suffices to show that
  there is a morphism
  from $(\nil \circ i, \Hist(c)_u)$ to $(\langle \id, (u \circ \cons_O)^\dagger \circ \obsv\rangle \circ i, \Pol(c))$ for each $u\colon O^+ \to A$. 
For each $u \colon O^+ \to A$, define a morphism
$\cmp_u \coloneqq \langle \last, \chi_u^\dagger \rangle \colon S^+ \to S \times A^{O^*}$
where $\chi_u \colon S^+ \times O^* \to A$ is the composite
\[
S^+ \times O^*
\xrightarrow{\obsv^+ \times \id}
O^+ \times O^*
\xrightarrow{\ext_O}
O^+
\xrightarrow{u}
A.
\]
We show that $\cmp_u\colon (\nil \circ i, \Hist(c)_u) \to (\langle \id, (u \circ \cons_O)^\dagger \circ \obsv\rangle \circ i, \Pol(c))$ in $\mathrm{Coalg}_I(FT)$.

First, $\cmp_u$ is an $FT$-coalgebra morphism
$\cmp_u \colon \Hist(c)_u \to \Pol(c)$
by the following commutative diagram:

\resizebox{.95\linewidth}{!}{\input{diag-hist-pol.tex}}

\noindent
where $f\colon S^+ \to A^{O^+}$
is the transpose of $S^+ \times O^+ \xrightarrow{\obsv^+ \times \id} O^+ \times O^+ \xrightarrow{\ext} O^+ \xrightarrow{u} A$.
The leftmost square commutes because, under the canonical isomorphism
\(A^{O^\ast}\cong A\times A^{O^+}\), the morphism \(\chi_u^\dagger\) corresponds to the pair
$\langle u\circ \obsv^+,\, f\rangle\colon S^+ \to A \times A^{O^+}$.

Next, we show that
$\cmp_u \circ \nil_S = \langle \id_S, (u \circ \cons)^\dagger \circ \obsv \rangle$, which induces that \(\cmp_u\) preserves the point.
For the first component, we use
$\last \circ \nil_S = \id_S$.
For the second component, by adjoint transposition, it suffices to show
$
\chi_u \circ (\nil_S \times \id_{O^*}) = u \circ \cons \circ (\obsv \times \id_{O^*})$.

This equality holds because of the naturality of $\nil$ and 
$\cons = \ext \circ (\nil \times \id)$.
\end{proof}
\ifarxiv

\subsection{Proof of~\Cref{lem:bel_pol}} \label{ap:proof_lem_correct_flat}
\begin{proof}
  Since $\flatten \circ \alpha = \id$, the following equations hold:
  \begin{align*}
    &F(\flatten_\obsv) \circ (c^\Bel)^\dagger \colon \domofslice{T}{O}{\obsv} \times A \to FTS\\
    &= F\mu \circ FT\iota_\obsv \circ F\alpha_\obsv \circ F\mu \circ \lambda_{TS} \circ T\ev \circ \str^T \circ (T\delta \circ \iota_\obsv) \times A  \\
    &= F\mu_S \circ \lambda_{TS} \circ T\delta^\dagger \circ \str^T \circ \iota_\obsv \times A.
  \end{align*}

  It further gives the following equations:
  \begin{align*}
    &F(\flatten_\obsv) \circ (c^\Bel)^\dagger \circ \eta^{\slicemonad{T}{O}}_\obsv \times A\\
    &= F\mu_S \circ \lambda_{TS} \circ T\delta^\dagger \circ \str^T \circ \eta_S \times \id &\text{by }\Cref{prop:iota} \\
    &= F\mu_S \circ \lambda_{TS} \circ T\delta^\dagger \circ \eta_{S \times A} \\ 
    &= \delta^\dagger &\text{by }\lambda_T \circ \eta_{FT} = F\eta_{T}.
  \end{align*}
\end{proof}

\subsection{Proof of~\Cref{lem:commute_delta_mu}} \label{ap:proof_correct}

\begin{proof}
  Let $e \coloneqq \langle \pi_1, \ev \circ (\obsv \times \id)\rangle\colon S \times A^{O^+} \to S \times A^{O^*}$,
  $\ol{e} \coloneqq \langle \pi_1, \ev \circ (\slicemonad{T}{O}\obsv \times \id)\rangle\colon \domofslice{T}{O}{\obsv} \times A^{O^+} \to \domofslice{T}{O}{\obsv} \times A^{O^*}$,
  and let $f$ be the canonical isomorphism $A^{O^*} \cong A \times A^{O^+}$.
  We also write $\ol{T}S$ for the object $\domofslice{T}{O}{\obsv}$.

1)
   By \Cref{lem:bel_pol}.\ref{item:bel_delta_eta},
   it follows that $\Pol(c) = F(\mu \circ T(\str^T \circ \iota \times \id)) \circ \Pol(\langle c^\Bel, \slicemonad{T}{O}\obsv \rangle) \circ (\eta^{\slicemonad{T}{O}}_\obsv \times \id)$
   by the following
   commutative diagram:

  \begin{equation} \label{eq:diag_pol_bel_1}
  \resizebox{1\linewidth}{!}{\input{diag-pol-bel-1.tex}}
  \end{equation}

  2)
Since $T\langle \id, \obsv \rangle \circ \iota_\obsv = \str^T \circ \langle \iota_\obsv, \slicemonad{T}{O}\obsv\rangle$,
the following equations hold:
\begin{align*}
  & \str^T \circ (\iota_\obsv \times A^{O^*}) \circ \ol{e} \colon \domofslice{T}{O}{\obsv} \times A^{O^+} \to T(S \times A^{O^*}) \\
  &= \str^T \circ (TS \times \ev) \circ (\langle \iota_\obsv, \slicemonad{T}{O}\obsv \rangle \times A^{O^+}) \\
&= T(S \times \ev) \circ \str^T_{S\times O, A^{O^+}} \circ (\str^T_{TS, O} \times A^{O^+}) \circ (\langle \iota_\obsv, \slicemonad{T}{O}\obsv\rangle \times A^{O^+}) \\
&= T(S \times \ev) \circ \str^T_{S\times O, A^{O^+}} \circ ((T\langle \id, \obsv\rangle \circ \iota_\obsv) \times A^{O^+}) \\
  &= T(e) \circ \str^T \circ (\iota_\obsv \times \id_{A^{O^+}}).
\end{align*}

  Then we obtain the following equations; they are summarized in the diagram below.
  \begin{align*}
    &F^2 \mu \circ F\lambda \circ FTF(\mu \circ T(\str^T \circ \iota \times A^{O^*})) \circ FT(\Pol(\langle c^\Bel, \slicemonad{T}{O}\obsv \rangle))   \\
    &= F^2 \mu \circ F\lambda \circ FTF\mu \circ (FT)^2(\str^T \circ (\iota \times A^{O^*}) \circ \ol{e}) \circ FT(\str^{FT} \circ ((c^\Bel)^\dagger \times A^{O^+}) \circ (\domofslice{T}{O}{\obsv} \times f))    \\ 
    &= F^2 \mu \circ F\lambda \circ FTF\mu \circ (FT)^2(T(e) \circ \str^T \circ (\iota \times A^{O^+})) \circ
    FT(\str^{FT} \circ ((c^\Bel)^\dagger \times A^{O^+}) \circ (\domofslice{T}{O}{\obsv} \times f)) \\
    &\qquad \text{ by naturality of }\str, \mu \text{ and commutativity of $\str$ with $\mu$}  \\
    &= F^2 \mu \circ F\lambda \circ (FT)^2(e) \circ FT\str^{FT} \circ FT((F\flatten_\obsv \circ (c^\Bel)^\dagger )\times A^{O^+}) \circ (\domofslice{T}{O}{\obsv} \times f) \\
    &\qquad \text{by \Cref{lem:bel_pol}.\ref{item:bel_delta} and naturality of }\mu, \lambda \\
    &= F^2 T(e) \circ F^2 \mu \circ F\lambda \circ FTF\str^{FT} \circ FT((F\mu_S \circ \lambda_{TS} \circ T\delta^\dagger \circ \str^T)\times A^{O^+}) \circ (\iota_\obsv  \times f) \\
    &\qquad \text{by naturality of }\str \\
    &= F^2 T(e) \circ F^2 \mu \circ F\lambda \circ FTF\str^{FT} \circ FTF(\mu \times A^{O^+}) \circ FT\str^F \circ FT((\lambda_{TS} \circ T\delta^\dagger \circ \str^T)\times A^{O^+}) \circ (\iota_\obsv  \times f) \\
    &\qquad \text{by compatibility of $\str$ with $\mu$} \\
    &= F^2 T(e) \circ F^2 \mu \circ F\lambda \circ FTF\mu \circ (FT)^2\str^T \circ FTF\str^T \circ FT\str^F \circ FT((\lambda_{TS} \circ T\delta^\dagger \circ \str^T)\times A^{O^+}) \circ (\iota_\obsv  \times f) \\
    &\qquad \text{by naturality of }\lambda \text{ and }F\str^T \circ \str^F \circ \lambda \times \id = \lambda \circ T(\str^F) \circ  \str^T  \\
    &= F^2 T(e) \circ F^2 \mu \circ F^2 T \mu \circ F\lambda \circ (FT)^2\str^T \circ FT\lambda \circ FT^2 \str^F \circ FT\str^T \circ FT((T\delta^\dagger \circ \str^T)\times A^{O^+}) \circ (\iota_\obsv  \times f) \\
    &\qquad \text{by }\mu_{S \times A^{O^+}} \circ T\mu_{S \times A^{O^+}} = \mu_{S \times A^{O^+}} \circ \mu_{T(S \times A^{O^+})}\text{ and naturality of }\lambda \\
    &= F^2 T(e) \circ F^2 \mu \circ F^2 \mu \circ F\lambda \circ FT\lambda \circ FT^2F\str^T \circ FT^2 \str^F \circ FT\str^T \circ FT((T\delta^\dagger \circ \str^T)\times A^{O^+}) \circ (\iota_\obsv  \times f) \\
    &\qquad \text{by compatibility of $\lambda$ with $\mu$} \\
    &= F^2 T(e) \circ F^2 \mu \circ F\lambda \circ F\mu \circ FT^2F\str^T \circ FT^2 \str^F \circ FT\str^T \circ FT((T\delta^\dagger \circ \str^T)\times A^{O^+}) \circ (\iota_\obsv  \times f) \\
    &\qquad \text{by naturality of }\mu, \lambda \\
    &= F^2\mu \circ F\lambda \circ FT(\Pol(c)) \circ F(\mu \circ T(\str^T \circ (\iota \times \id))).
  \end{align*}

\resizebox{\linewidth}{!}{\input{diag-pol-bel-big-n-rotate.tex}}
\end{proof}

\subsection{Omitted proof for \Cref{thm:belief-correctness}}
\label{ap:diag_correctness_n}

\begin{lem}
  For each $n \geq 1$, the following equation holds:
  \[
  \tau^{n} \circ (FT)^{n} \bot \circ FT(\Pol(c)^{n-1}) \circ F(\mu \circ T(\str^T \circ (\iota_\obsv \times \id))) 
  = \tau^{n} \circ (FT)^{n} \bot \circ FT(\Pol(\langle c^\Bel, \slicemonad{T}{O}\obsv \rangle)^{n-1}).
  \]
\end{lem}
\begin{proof}
  We prove it by induction on $n$.

  For the base case $n=1$,
  \begin{align*}
  &\tau \circ FT \bot \circ F(\mu \circ T(\str^T \circ (\iota_\obsv \times \id)))    \\
  &= \rho \circ F\sigma \circ FT\bot \circ F(\mu \circ T(\str^T \circ (\iota_\obsv \times \id)))  &\text{by definition of }\tau\\
  &= \rho \circ F\sigma \circ FT\bot \circ FT(\str^T \circ (\iota_\obsv \times \id)) \\
  &= \tau \circ FT\bot &\text{since $\bot$ is preserved under precomposition.}
  \end{align*}
  The third equality holds because for each $X \in \cat{C}$,
  \begin{align*}
    \sigma \circ T\bot_{X} \circ \mu_{X} 
    &= \sigma \circ \mu_{T\Omega} \circ T^2 \bot_{X} &\text{by naturality of $\mu$} \\
    &= \sigma \circ T\sigma \circ T^2 \bot_{X} &\text{since }\sigma \text{ is an EM-algebra} \\
    &= \sigma \circ T\bot_{TX} &\text{by $\bot_{TX} = \sigma \circ T\bot_X$.}
  \end{align*}

  For the step case, we assume that \eqref{eq:ast_n} for $n$ holds. 
  The commutativity of the following diagram
  yields the desired equality for $n+1$.

  \resizebox{1\linewidth}{!}{\input{diag-pol-bel-n.tex}}

  \noindent
  The whole diagram commutes
  by 
  $\sigma \circ T\rho = \rho \circ F(\sigma) \circ \lambda$,
$\sigma \circ \mu = \sigma \circ T\sigma$,
compatibility of $\lambda$ with $\mu$,
and
  \Cref{lem:commute_delta_mu}.

  Therefore, \eqref{eq:ast_n} holds for all $n \geq 1$, and the statement follows.

\end{proof}

\subsection{The Details of~\cref{ex:nondetAssumption}}\label{ap:nondetAssump}

It is straightforward to see that $\rho$ and $\sigma$ are monotone, and $\sigma$ is an Eilenberg-Moore algebra. 
We see that $\sigma \circ \pefunc(\rho) = \rho \circ (\sigma + \{\checkmark\}) \circ \lambda_\bset$: For any $U\in \pefunc(\bset + \{\checkmark\})$,  
\begin{align*}
  \big(\sigma \circ \pefunc(\rho)\big)(U) &= \land \big(\pefunc(\rho)(U)\big) = \begin{cases*}
    \bfalse \text{ if } \bfalse\in U,\\
    \btrue \text{ otherwise}. 
  \end{cases*}\\
  \big(\rho \circ (\sigma + \{\checkmark\}) \circ \lambda_\bset\big)(U) &= \begin{cases*}
    \big(\rho \circ (\sigma + \{\checkmark\})\big)(\checkmark) \text{ if } U = \{\checkmark\},\\
    \big(\rho \circ (\sigma + \{\checkmark\})\big)(U\backslash\{\checkmark\}) \text{ otherwise}
  \end{cases*} = 
  \begin{cases*}
    \btrue \text{ if } U = \{\checkmark\},\\
    \bfalse \text{ if $\bfalse \in U\backslash \{\checkmark\}$},\\
    \btrue \text{ otherwise }\\
  \end{cases*}\\
  &= \begin{cases*}
    \bfalse \text{ if } \bfalse\in U,\\
    \btrue \text{ otherwise}. 
  \end{cases*}
\end{align*}
It is easy to see that the second assumption in $(\star)$ is satisfied. 

\subsection{Proof of~\Cref{prop:partial-information-upper-bound}} \label{ap:proof_fully_observable}

\begin{proof}
Note that
$(\id_S,\obsv)\colon
\langle \delta,\obsv\rangle \to \langle \delta,\id_S\rangle$
is a morphism in
$\CoalgOctr{FT(\_)^A}$.
Hence
$\Pol(\id_S,\obsv)\colon
\Pol(i, \langle \delta,\obsv\rangle)
\to
\Pol(i, \langle \delta,\id_S\rangle)$
is an $FT$-coalgebra morphism.
Consider an arbitrary morphism $u\colon O^+ \to A$ and $n \in \mathbb{N}$.
Applying~\cref{prop:Phi-coalg-morphism}, we obtain
$\Phi_{\Pol(i, \langle \delta,\obsv\rangle)}^n(\bot) = \Phi_{\Pol(i, \langle \delta,\id_S\rangle)}^n(\bot)
\circ (\id \times A^{\obsv^*})$.
Moreover, 
since $\cons \circ (\obsv \times \obsv^*) = \obsv^+ \circ \cons\colon S \times S^* \to O^+$,
 we have $A^{\obsv^*} \circ (u \circ \cons)^\dagger \circ \obsv = (u \circ \obsv^+ \circ \cons)^\dagger\colon S \to A^{S^*}$.
 Therefore,
\begin{align*}
&\Phi_{\Pol(i, \langle \delta,\obsv\rangle)}^n(\bot)
\circ \langle \id, (u \circ \cons)^\dagger \circ \obsv \rangle \circ i \\
&= \Phi_{\Pol(i, \langle \delta,\id_S\rangle)}^n(\bot)
\circ (\id \times A^{\obsv^*})
\circ \langle \id, (u \circ \cons)^\dagger \circ \obsv \rangle \circ i \\
&= \Phi_{\Pol(i, \langle \delta,\id_S\rangle)}^n(\bot)
\circ \langle \id, (u \circ \obsv^+ \circ \cons)^\dagger \rangle \circ i
\le
\Vc{i, \langle \delta,\id_S\rangle} .
\end{align*}

If $\obsv$ is split mono, then there exists a morphism $g\colon O \to S$ such that $g \circ \obsv = \id_S$.
Then $(\id_S, g)\colon \langle \delta,\id_S\rangle \to \langle \delta,\obsv\rangle$ is a morphism in $\CoalgOctr{FT(\_)^A}$, and hence $\Pol(\id_S, g)\colon \Pol(i, \langle \delta,\id_S\rangle) \to \Pol(i, \langle \delta,\obsv\rangle)$ is an $FT$-coalgebra morphism.
By a similar argument as above, 
for each $u\colon S^+ \to A$ and $n \in \mathbb{N}$, 
we have
$\Phi_{\Pol(i, \langle \delta,\id_S\rangle)}^n(\bot) \circ \langle \id, (u \circ \cons)^\dagger \rangle \circ i \le \Vc{i, \langle \delta, \obsv\rangle}$.
\end{proof}

\subsection{Proof of \Cref{prop:partial_upper_hist_new}} \label{ap:proof_partial_upper_hist_new}

\begin{proof}
  The inequality 
$\Hc{i, \langle \delta, \obsv\rangle} \leq \Hc{i, \langle \delta, \id\rangle}$ follows from
  \Cref{prop:partial-information-upper-bound} and \Cref{prop:history-semantics}.

  Because 
  $m_u\colon (i_u, \delta_u) \rightarrowtail (\nil \circ i, \Hist(\langle \delta, \id\rangle)_u)$
  and $f_u \circ \obsv^+ \circ m_u = m_u$,
  it follows that 
  \[m_u\colon (i_u, \delta_u) \rightarrowtail (\nil \circ i, \Hist(\langle \delta, \id\rangle)_{u \circ f_u \circ \obsv^+}).\]
  Moreover,
  by \Cref{lem:subcoalg_u}, the morphism $(\id, \obsv^+)\colon \Hist(\langle \delta, \id\rangle) \to \Hist(c)$ in 
  $\CoalgOcov{FT(\_)^A}{}$ yields the morphism 
  \[\id\colon \Hist(\langle \delta, \id\rangle)_{u \circ f_u \circ \obsv^+} \to \Hist(c)_{u \circ f_u}\] in $\Coalg{FT}$.
  Thus we have $m_u\colon (i_u, \delta_u) \rightarrowtail (\nil \circ i, \Hist(c)_{u \circ f_u})$.
  By \Cref{prop:Phi-coalg-morphism}, we obtain the following:
  \begin{align*}
    \Hc{i, \langle \delta, \id\rangle}
    &= \bigvee_{u\colon S^+ \to A,\; n \in \mathbb{N}} \Phi_{\Hist(\langle \delta, \id\rangle)_u}^n(\bot) \circ \nil \circ i\\
    &= \bigvee_{u\colon S^+ \to A,\; n \in \mathbb{N}} \Phi_{\delta_u}^n(\bot) \circ i_u \\
    &= \bigvee_{u\colon S^+ \to A,\; n \in \mathbb{N}} \Phi_{\Hist(c)_{u \circ f_u}}^n(\bot) \circ \nil \circ i  \\
    &\leq \Hc{i, \langle \delta, \obsv\rangle}.
  \end{align*}
\end{proof}

\subsection{Omitted Proofs in~\cref{sec:reachable}}
\label{sec:omitReachable}

\begin{proposition}
  \label{prop:decidable}
  Let  $c = \langle \delta,\obsv\rangle\colon S \to (\pefunc(S) + \{\checkmark\})^A \times O$ and $A$ and $S$ be finite sets. 
  Computing the objective $\Vc{i, c}$  is decidable. 
\end{proposition}
\begin{proof}
  Since its belief coalgebra is still finite, and we can assume that this belief coalgebra is fully observable by restricting to reachable parts, 
  it suffices to show that computing the objective $\Vc{i, c'}$ for a coalgebra  $c' = \langle \delta',\id_S'\rangle\colon S' \to (\pefunc(S') + \{\checkmark\})^A \times S'$ with a finite set $S'$ is decidable. 
  Since it is fully observable, we can see that the objective $\Vc{i, c'}$ is 
  the composition of $i\colon I \to S'$ with 
  the least fixed point of the following operator $\Psi\colon \bset^{S'}\rightarrow \bset^{S'}$: 
  \begin{align*}
    \Psi(f)(s) \coloneqq \begin{cases*}
      \btrue \qquad\text{ if }\exists a\in A\text{ s.t. } \delta'(s)(a) = \checkmark\text{ or } f\big(\delta'(s)(a)\big) = \{\btrue\},\\
      \bfalse\qquad\text{ otherwise. }
    \end{cases*}
  \end{align*}
  Since $\Psi$ is $\omega$-continuous and the domain $\bset^{S'}$ is a finite complete lattice, the least fixed point $\Vc{i, c'}$ of $\Psi$ can be obtained in finite time by Kleene iteration. 
\end{proof}

\else
\fi

\end{document}

%% file: diag-hist-pol.tex
\begin{tikzcd}
	{S^+} && {FTS^A \times A \times S^+} & {FT(S \times S^+)} & {FT(S^+ \times S)} & {FT(S^+)} \\
	{S \times A^{O^*}} && {FTS^A \times A \times A^{O^+}} & {FT(S \times A^{O^+})} && {FT(S \times A^{O^*})}
	\arrow["{\langle \delta \circ \mathrm{last}_{\mathrm{id}}, u\circ \obsv^+, \mathrm{id}\rangle}", from=1-1, to=1-3]
	\arrow["{{\mathrm{cmp}_u}}"', from=1-1, to=2-1]
	\arrow["{{\mathrm{st} \circ \mathrm{ev} \times \mathrm{id}_{S^+}}}", from=1-3, to=1-4]
	\arrow["{{FTS^A \times A \times f}}"{description}, from=1-3, to=2-3]
	\arrow["\cong", from=1-4, to=1-5]
	\arrow["{{FT(S \times f)}}", from=1-4, to=2-4]
	\arrow["{{FT\mathrm{ext}}}", from=1-5, to=1-6]
	\arrow["{{FT(\mathrm{cmp}_u)}}", from=1-6, to=2-6]
	\arrow["{\langle \delta, \cong \rangle}", from=2-1, to=2-3]
	\arrow["{{\mathrm{st} \circ \mathrm{ev} \times \mathrm{id}}}", from=2-3, to=2-4]
	\arrow["{{FT\langle \pi_1, \mathrm{ev} \circ \mathrm{obs} \times \mathrm{id}\rangle}}", from=2-4, to=2-6]
\end{tikzcd}

%% file: diag-pol-bel-1.tex
\begin{tikzcd}
	{S \times A^{O^*}} & {S \times A \times A^{O^+}} & {FTS \times A^{O^+}} & {FT(S \times A^{O^+})} & {FT(S \times A^{O^*})} \\
	{\overline{T}S \times A^{O^*}} & {\overline{T}S \times A \times A^{O^+}} & {FT\overline{T}S \times A^{O^+}} & {FT(\overline{T}S \times A^{O^+})} & {FT(\overline{T}S \times A^{O^*})}
	\arrow["\cong", from=1-1, to=1-2]
	\arrow["{\eta^{\overline{T}}_\mathrm{obs} \times \mathrm{id}}"', from=1-1, to=2-1]
	\arrow["{\delta^\dagger \times \mathrm{id}}", from=1-2, to=1-3]
	\arrow["{\eta^{\overline{T}}_\mathrm{obs} \times \mathrm{id}}"{description}, from=1-2, to=2-2]
	\arrow["{\mathrm{st}}", from=1-3, to=1-4]
	\arrow["{FT(e)}", from=1-4, to=1-5]
	\arrow["\cong", from=2-1, to=2-2]
	\arrow["{(\delta^\mathrm{Bel})^\dagger \times \mathrm{id}}"', from=2-2, to=2-3]
	\arrow["{F(\mathrm{flat}_\mathrm{obs}) \times \mathrm{id}}"{description}, from=2-3, to=1-3]
	\arrow["{\mathrm{st}}", from=2-3, to=2-4]
	\arrow["{F\mu \circ FT(\mathrm{st} \circ \iota \times \mathrm{id})}"{description}, from=2-4, to=1-4]
	\arrow["{FT(\overline{e})}", from=2-4, to=2-5]
	\arrow["{F\mu \circ FT(\mathrm{st} \circ \iota \times \mathrm{id})}"', from=2-5, to=1-5]
\end{tikzcd}

%% file: diag-pol-bel-big-n-rotate.tex
\begin{tikzcd}
	{F^2T(S \times A^{O^*})} & {F^2T(S \times A^{O^+})} &&& \\
	{F^2T^2(S \times A^{O^*})} & {F^2T^2(S \times A^{O^+})} \\
	{(FT)^2(S \times A^{O^*})} & {(FT)^2(S \times A^{O^+})} \\
	{(FT)^2T(S \times A^{O^*})} & {(FT)^2T(S \times A^{O^+})} \\
	{(FT)^2(T S \times A^{O^*})} &&& {F^2T^3( S \times A^{O^+})} \\
	{(FT)^2(\overline{T} S \times A^{O^*})} &&& {(FT)^2T( S \times A^{O^+})} & {(FT)^2(S \times A^{O^+})} \\
	{(FT)^2(\overline{T} S \times A^{O^+})} & {(FT)^2(T S \times A^{O^+})} & {FTF(T S \times A^{O^+})} & {(FT)^2(T S \times A^{O^+})} \\
	{FT(FT\overline{T} S \times A^{O^+})} & {FT(FTT S \times A^{O^+})} & {FT(FT S \times A^{O^+})} & {FTF(T^2 S \times A^{O^+})} & {FT^2FT(S \times A^{O^+})} \\
	&& {FT(FT^2 S \times A^{O^+})} & {FT^2F(TS \times A^{O^+})} \\
	&& {FT(TFTS \times A^{O^+})} & {FT^2(FTS \times A^{O^+})} & {FT(FTS \times  A^{O^+})} \\
	{FT(\overline{T}S \times A  \times  A^{O^+})} & {FT(TS \times A \times A^{O^+})} & {FT(T(S \times A) \times A^{O^+})} & {FT^2(S \times A \times A^{O^+})} & {FT(S \times A \times A^{O^+})} \\
	{FT(\overline{T}S \times A^{O^*})} & {FT(TS \times A^{O^*})} & {FT^2(S \times A^{O^*})} && {FT(S \times A^{O^*})}
	\arrow["{F^2Te}", from=1-2, to=1-1]
	\arrow["{F^2\mu}", from=2-1, to=1-1]
	\arrow["{F^2\mu}"', from=2-2, to=1-2]
	\arrow["{F^2T^2e}", from=2-2, to=2-1]
	\arrow["{{F\lambda}}", from=3-1, to=2-1]
	\arrow["{{F\lambda}}", from=3-2, to=2-2]
	\arrow["{(FT)^2e}", from=3-2, to=3-1]
	\arrow["{{FTF\mu}}", from=4-1, to=3-1]
	\arrow["{FTF\mu}", from=4-2, to=3-2]
	\arrow["{(FT)^2Te}", from=4-2, to=4-1]
	\arrow["{{(FT)^2 \mathrm{st}^T}}", from=5-1, to=4-1]
	\arrow["{F^2T\mu}"{description}, from=5-4, to=2-2]
	\arrow["{F^2\mu}"{description}, curve={height=12pt}, from=5-4, to=2-2]
	\arrow["{(FT)^2 (\iota \times \mathrm{id})}", from=6-1, to=5-1]
	\arrow["{FTF\mu}"{description}, from=6-4, to=3-2]
	\arrow["{F\lambda}"', from=6-4, to=5-4]
	\arrow["{F\lambda}"{description}, shift right=3, curve={height=30pt}, from=6-5, to=2-2]
	\arrow["{{(FT)^2 (\overline{e})}}", from=7-1, to=6-1]
	\arrow["{(FT)^2 (\iota \times \mathrm{id})}", from=7-1, to=7-2]
	\arrow["{(FT)^2\mathrm{st}^T}", from=7-2, to=4-2]
	\arrow["{{FTF\mathrm{st}^T}}"{description}, curve={height=12pt}, from=7-3, to=3-2]
	\arrow["{(FT)^2\mathrm{st}^T}", from=7-4, to=6-4]
	\arrow["{{FT\mathrm{st}^{FT}}}", from=8-1, to=7-1]
	\arrow["{FT(FT\iota\times\mathrm{id})}", from=8-1, to=8-2]
	\arrow["{FT\mathrm{st}^{FT}}", from=8-2, to=7-2]
	\arrow["{FT(F\mu \times \mathrm{id})}"', from=8-2, to=8-3]
	\arrow["{{FT\mathrm{st}^F}}", from=8-3, to=7-3]
	\arrow["{FTF(\mu \times \mathrm{id})}"{description}, from=8-4, to=7-3]
	\arrow["{FTF\mathrm{st}^T}", from=8-4, to=7-4]
	\arrow["{FT\lambda}"{description}, from=8-5, to=6-4]
	\arrow["{{FT(F\mu \times \mathrm{id})}}", from=9-3, to=8-3]
	\arrow["{{FT\mathrm{st}^F}}"{description}, from=9-3, to=8-4]
	\arrow["{FT\lambda}"'{pos=0.6}, shift right, curve={height=30pt}, from=9-4, to=7-4]
	\arrow["{FT^2F\mathrm{st}^T}"', from=9-4, to=8-5]
	\arrow["{{FT(\lambda \times \mathrm{id})}}", from=10-3, to=9-3]
	\arrow["{{FT\mathrm{st}^T}}", from=10-3, to=10-4]
	\arrow["{FT^2\mathrm{st}^F}", from=10-4, to=9-4]
	\arrow["{{F\mu}}"', from=10-4, to=10-5]
	\arrow["{FT\mathrm{st}^{FT}}"', shift right=3, curve={height=40pt}, from=10-5, to=6-5]
	\arrow["{{FT(\langle c^\mathrm{Bel}, \overline{T}\mathrm{obs}\rangle^\dagger \times \mathrm{id})}}", from=11-1, to=8-1]
	\arrow["{{FT(\iota \times \mathrm{id})}}", from=11-1, to=11-2]
	\arrow["{{FT(\mathrm{st}^T \times \mathrm{id})}}"', from=11-2, to=11-3]
	\arrow["{{FT(T\delta^\dagger \times \mathrm{id})}}", from=11-3, to=10-3]
	\arrow["{{FT\mathrm{st}^T}}"', from=11-3, to=11-4]
	\arrow["{{FT^2(\delta^\dagger \times \mathrm{id})}}", from=11-4, to=10-4]
	\arrow["{{F\mu}}"', from=11-4, to=11-5]
	\arrow["{{FT(\delta^\dagger \times \mathrm{id})}}"', from=11-5, to=10-5]
	\arrow["{{FT(\mathrm{id} \times f)}}", from=12-1, to=11-1]
	\arrow["{{FT(\iota \times \mathrm{id})}}"', from=12-1, to=12-2]
	\arrow["{{FT(\mathrm{id} \times f)}}", from=12-2, to=11-2]
	\arrow["{{FT\mathrm{st}^T}}"', from=12-2, to=12-3]
	\arrow["{{FT^2(\mathrm{id} \times f)}}"{description}, from=12-3, to=11-4]
	\arrow["{{F\mu}}"', from=12-3, to=12-5]
	\arrow["{{FT(\mathrm{id} \times f)}}"', from=12-5, to=11-5]
\end{tikzcd}

%% file: diag-pol-bel-n.tex
\begin{tikzcd}
	{FT(\overline{T}S \times A^{O^*})} & {(FT)^2(\overline{T}S \times A^{O^*})} & {(FT)^{n+1}(\overline{T}S \times A^{O^*})} & {(FT)^{n+1}\Omega} &&& \Omega \\
	{FT(S \times A^{O^*})} & {(FT)^2(S \times A^{O^*})} & {(FT)^{n+1}(S \times A^{O^*})} & {F^2T^2(FT)^{n-1}\Omega} & {F^{n+1}T^{n+1}\Omega} & {F^{n+1}T^2\Omega} & {F^{n+1}\Omega} \\
	& {F^2T^2(S \times A^{O^*})} & {F^2T^2(FT)^{n-1}(S \times A^{O^*})} & {F^2T(FT)^{n-1}\Omega} & {F^{n+1}T^{n}\Omega} & {F^{n+1}T\Omega} \\
	& {F^2T(S \times A^{O^*})} & {F^2T(FT)^{n-1}(S \times A^{O^*})}
	\arrow["{FT(\Pol(\langle \delta^\Bel, \slicemonad{T}{O}\obsv \rangle))}"'{pos=0}, from=1-1, to=1-2]
	\arrow["{FT(\Pol(\langle \delta^\Bel, \slicemonad{T}{O}\obsv \rangle)^n)}", curve={height=-18pt}, from=1-1, to=1-3]
	\arrow["{F(\mu \circ T(\mathrm{st} \circ \iota \times \mathrm{id}))}"{description}, from=1-1, to=2-1]
	\arrow["{(FT)^2(\Pol(\langle \delta^\Bel, \slicemonad{T}{O}\obsv \rangle)^{n-1})}"', from=1-2, to=1-3]
	\arrow["{FTF(\mu \circ T(\mathrm{st} \circ \iota \times \mathrm{id}))}"{description}, from=1-2, to=2-2]
	\arrow["{(FT)^{n+1}\bot}", from=1-3, to=1-4]
	\arrow["{\tau^{n+1}}", from=1-4, to=1-7]
	\arrow["{F\lambda}"', from=1-4, to=2-4]
	\arrow["{FT(\Pol(c))}"', from=2-1, to=2-2]
	\arrow["{(FT)^2\Pol(c)^{n-1}}", from=2-2, to=2-3]
	\arrow["{ F\lambda}"{description}, from=2-2, to=3-2]
	\arrow["{(FT)^{n+1}\bot}"{description}, from=2-3, to=1-4]
	\arrow["{F\lambda}"', from=2-3, to=3-3]
	\arrow["{\vec{\lambda}}", from=2-4, to=2-5]
	\arrow["{F^2\mu}"{description}, from=2-4, to=3-4]
	\arrow["{F^{n+1}T\sigma^n}", from=2-5, to=2-6]
	\arrow["{F^{n+1}\mu}", from=2-5, to=3-5]
	\arrow["{F^{n+1}\sigma^2}"', from=2-6, to=2-7]
	\arrow["{F^{n+1}\mu}", from=2-6, to=3-6]
	\arrow["{\rho^{n+1}}"', from=2-7, to=1-7]
	\arrow["{F^2 \mu}"', from=3-2, to=4-2]
	\arrow["{F^2T^2(FT)^{n-1}\bot}"{description}, from=3-3, to=2-4]
	\arrow["{F^2\mu}"{description}, from=3-3, to=4-3]
	\arrow["{\vec{\lambda}}", from=3-4, to=3-5]
	\arrow["{F^{n+1}T^2\sigma^n}", from=3-5, to=3-6]
	\arrow["{F^{n+1}\sigma}"', curve={height=6pt}, from=3-6, to=2-7]
	\arrow["{F^2T\Pol(c)^{n-1}}"', from=4-2, to=4-3]
	\arrow["{F^2T(FT)^{n-1}\bot}"{description}, from=4-3, to=3-4]
\end{tikzcd}